\documentclass[onecolumn, conference]{IEEEtran}

\usepackage{amssymb,amsthm}
\usepackage{amsmath,amssymb}
\usepackage{bm,amsfonts}
\usepackage{multirow}
\usepackage{comment}
\usepackage{graphicx}

\DeclareMathOperator*{\argmax}{\rm argmax}

\newcommand{\so}{\mathrm{o}}
\newcommand{\lo}{\mathrm{O}}
\newcommand{\ep}{\epsilon}
\newcommand{\de}{\delta}

\newcommand{\si}{\sigma}
\newcommand{\id}{ \mbox{\rm 1}\hspace{-0.5em}\mbox{\rm \small l\,}}

\newcommand{\idx}[1]{\id\!\left[#1\right]}

\newcommand{\E}{\mathrm{E}}

\newcommand{\n}{\nonumber}
\newcommand{\nn}{\nonumber\\}

\newcommand{\uR}{\underline{\mathbb{R}}}

\newtheorem{lemma}{Lemma}
\newtheorem*{lemmat}{Lemma 1}
\newtheorem{theorem}{Theorem}
\newtheorem{proposition}{Proposition}
\newtheorem{definition}{Definition}
\newtheorem{remark}{Remark}
\newcommand{\fqed}{\hfill $\blacksquare$ \par}

\newcommand{\la}{\lambda}
\newcommand{\La}{\mathrm{\Lambda}}
\newcommand{\e}{\mathrm{e}}
\newcommand{\De}{\Delta}

\newcommand{\bibun}[2]{\frac{\rd #1}{\rd #2}}
\newcommand{\com}{\,,}
\newcommand{\per}{\,.}

\newcommand{\bbR}{\mathbb{R}}
\newcommand{\calX}{\mathcal{X}}
\newcommand{\calY}{\mathcal{Y}}

\newcommand{\since}[1]{\quad\left(\mbox{#1}\right)}

\newcommand{\im}{\mathrm{i}}

\newcommand{\bmX}{\bm{X}}
\newcommand{\bmY}{\bm{Y}}
\newcommand{\bmx}{\bm{x}}
\newcommand{\bmy}{\bm{y}}

\newcommand{\bZ}{\bar{Z}}

\newcommand{\al}{\alpha}

\newcommand{\bibunZ}[1]{Z^{(#1)}}

\newcommand{\pp}{p_+}
\newcommand{\pz}{p_0}

\newcommand{\gapxia}{b_0}
\newcommand{\gapxi}{b_1}

\newcommand{\cumu}{{Z_{\mathrm{a}}}}

\newcommand{\cumui}[1]{{Z_{\mathrm{a},i}}}

\newcommand{\PRC}{P_{\mathrm{RC}}}
\newcommand{\sqn}{\sqrt{n}}
\newcommand{\Rcrit}{R_{\mathrm{crit}}}
\newcommand{\gsin}{g^{(\mathrm{s})}}
\newcommand{\tgsin}{\tilde{g}^{(\mathrm{s})}}
\newcommand{\psin}{\psi^{(\mathrm{s})}}
\newcommand{\tpsin}{\tilde{\psi}^{(\mathrm{s})}}

\newcommand{\Ex}[1]{\E\left[#1\right]}
\newcommand{\pax}[1]{\left(#1\right)}
\newcommand{\sqx}[1]{\left[#1\right]}
\newcommand{\brx}[1]{\left\{#1\right\}}

\newcommand{\calJ}{\mathcal{J}}

\newcommand{\bbZ}{\mathbb{Z}}
\newcommand{\bbQ}{\mathbb{Q}}
\newcommand{\tZ}{\tilde{Z}}
\newcommand{\calZ}{\mathcal{Z}}
\newcommand{\calW}{\mathcal{W}}
\newcommand{\rd}{\mathrm{d}}

\title{
Comprehensive Analysis on
Exact Asymptotics of \\ Random Coding Error Probability
}%
\author{\authorblockN{Junya Honda}
\authorblockA{Graduate School of Frontier Sciences,
The University of Tokyo\\
 Kashiwa-shi Chiba 277--8561, Japan\\
Email: honda@it.k.u-tokyo.ac.jp}
}

\begin{document}
\maketitle
\allowdisplaybreaks[4]

\begin{abstract}
This paper
considers error probabilities of random codes for
memoryless channels in the fixed-rate regime.
Random coding is a fundamental
scheme to achieve the channel capacity
and
many studies have been conducted for
the asymptotics of the decoding error probability.
Gallager derived the exact asymptotics (that is, a bound with asymptotically vanishing relative error)
of the error probability for fixed rate below the critical rate.
On the other hand,
exact asymptotics for rate above the critical rate 
has been unknown except for symmetric channels (in the strong sense)
and strongly nonlattice channels.
This paper derives the exact asymptotics
for general memoryless channels covering
all previously unsolved cases.
The analysis reveals that
strongly symmetric channels
and strongly nonlattice channels
correspond to two extreme cases
and
the expression of the asymptotics
is much complicated
for general channels.
\end{abstract}


\section{Introduction}

Random coding is a fundamental
scheme in many problems of information theory
and asymptotically achieves the capacity in channel coding.
This code
is also important
in the finite block length regime
to clarify the achievable performance of channel codes.
For this purpose Polyanskiy \cite{second_polyanskiy} and Hayashi \cite{second_hayasi}
considered random codes with varying coding rate for fixed error probability
and revealed that loss in the coding rate from the capacity
is $\lo(1/\sqrt{n})$
for the block length $n$.

Whereas these bounds are convenient
for evaluate the error probability with small absolute error,
it is sometimes useful to evaluate the error probability
with small {\it relative} error
when the acceptable error probability is very small.
This line of work
is closely related to the theory of error exponent,
which considers exponential decay of the error probability
for fixed rate $R$.
Gallager \cite{gallager_map} derived an upper bound of the
error probability of random coding
called a random coding union bound.
It is shown that the use of union bound
does not worsen the exponent
for coding rates below the critical rate \cite{gallager_map}
and above the critical rate \cite{dyachkov}.

As a higher-order analysis for the error exponent,
there are many studies 
to evaluate the random coding error probability $P_{\mathrm{RC}}(n)$
with vanishing relative error
for fixed coding rate $R$ and block length $n$
for memoryless channels.
Dobrushin \cite{dobrushin}
showed that the random coding error probability
is written in a form $\Theta(n^{-a(R)}\e^{-nE(R)})$
for discrete symmetric channels in the strong sense that
each row and the column of the transition probability matrix
are permutations of the others.
They also derived the specific value of
$\lim_{n\to\infty}n^{a(R)}\e^{E(R)}P_{\mathrm{RC}}(n)$
for the nonlattice case (defined later)
and noted that
the limit does not exist for some cases.

For general
class of discrete memoryless channels,
Gallager \cite{gallager_tight} showed that
the upper bound derived in \cite{gallager_map}
is also the lower bound
with vanishing relative error
for rate below the critical rate.
Altu\u{g} and Wagner \cite{altug_journal} corrected his result
for singular channels.
They, 
and Scarlett et al.~\cite{scarlett},
also derived
upper bounds of the error probability
for general rate $R$.
However
these bounds, denoted by $\hat{P}(n)$,
do not assure $\hat{P}(n)/P_{\mathrm{RCU}}(n)=\Omega(1)$
although $\hat{P}(n)/P_{\mathrm{RCU}}(n)=\lo(1)$ is proved.


Honda \cite{exact_isit} derived a framework to evaluate
the random coding error probability for general (possibly nondiscrete)
nonsingular memoryless channels.
He introduced a two-dimensional random variable,
which will be denoted by $(Z(\eta),Z'(\eta))$ or $(Z_0,Z_1)$ for short,
and showed that $\lim_{n\to\infty}\E[f_n(\bZ_0,\bZ_1)]/P_{\mathrm{RC}}(n)=1$
for some function $f_n$,
where $(\bZ_0,\bZ_1)$ is the empirical mean of $n$ i.i.d.~copies of $(Z_0,Z_1)$.
Thus, we can obtain an explicit representation of $\PRC(n)$
if $\E[f_n(\bZ_0,\bZ_1)]$
is approximated appropriately.
It is known that
the error of
normal approximation
of $(\bZ_0,\bZ_1)$
becomes large
if Cram\'er's condition is not satisfied, or equivalently,
if $(Z_0,Z_1)$ is distributed over a lattice
or a set of parallel lines with equal interval.
In these cases the analysis becomes much complicated
and Honda \cite{exact_isit} only derived an explicit representation of
$\E[f_n(\bZ_0,\bZ_1)]$ for the case that Cram\'er's condition is satisfied.
For
continuous channels such as Gaussian channels
lattice distributions do not appear and
a higher-order analysis is given in \cite{awgn_finite}.

In this paper we derive simple representation of $\E[f_n(\bZ_0,\bZ_1)]$,
or equivalently $\PRC(n)$,
for general $(Z_0,Z_1)$ including the case that
$(Z_0,Z_1)$ is distributed over a lattice
or parallel lines, which is the last region where
the exact asymptotics of the random coding error probability
has been unknown for singular channels.
Our analysis reveals that
strongly symmetric channels considered in \cite{dobrushin} belong to the degenerate case that
$Z_1$ is a linear deterministic function of
$Z_0$ and the asymptotic form of the error probability becomes much simpler.
We also derive the exact asymptotics for singular channels
by applying the same techniques.
Thus our analysis covers
all previously unknown cases
in the evaluation of random coding error probability
with vanishing relative gap for fixed rate $R$.

The main difficulty of the derivation is that
the required precision for the evaluation
of $(Z_0,Z_1)$ is not ``isotropic''.
More precisely,
$\E[f_n(\bZ_0,\bZ_1)]$
depends on the behavior of $\bZ_0$
in $\so(1/\sqn)$ precision
whereas
$\E[f_n(\bZ_0,\bZ_1)]$ has rough dependence
on $\bZ_1$
and $\so(1/\sqn)$ precision for $\bZ_2$ does not lead
to a simple expression.
Based on this observation,
we start with
local limit theorem for $(\bZ_0,\bZ_1)$ with
$\so(1/\sqrt{n})$ precision in both directions
and ``blur'' the distribution function only in $\bZ_1$ direction.

\section{Preliminary}
We consider a memoryless channel
with input alphabet $\calX$ and output alphabet $\calY$.
The output distribution for input $x\in\calX$
is denoted by $W(\cdot |x)$.
Let $X\in \calX$ be a random variable with distribution $P_X$
and $Y\in\calY$ follow
$W(\cdot|X)$ given $X$.
$X'$ is a random variable
with the same distribution as $X$ and independent of $(X,Y)$.
$(\bmX,\bmY,\bmX')=((X_1,\cdots,X_n),\,(Y_1,\cdots,Y_n),\,(X_1',\cdots,X_n'))$ denotes
$n$ independent copies of $(X,Y,X')$.

We assume that there exists a base measure $Q$
such that $W(\cdot|x)$ is absolutely continuous with respect to $Q$
for all $x$.
Under this assumption,
we also use
$W(y|x)$
to denote
the Radon-Nikodym derivative $\frac{\rd W(\cdot|x)}{\rd Q}(y)$
by a slight abuse of notation.
Since the density satisfies $W(Y|X)>0$ almost surely,
the log likelihood ratio
\begin{align}
\nu(X,Y,x')=\log \frac{W(Y|x')}{W(Y|X)}\in [-\infty,\infty)\n
\end{align}
is well-defined almost surely
for any $x'\in\calX$.
We assume that the mutual information is finite, that is,
$I(X;Y)=\E_{XY}[-\log \E_{X'}[\e^{\nu(X,Y,X')}]]<\infty$.


We consider the error probability of
a random code such that each element of codewords
$(\bmX_1,\cdots,\allowbreak
\bmX_M)\in \calX^{n\times M}$
is generated independently from distribution $P_X$.
The coding rate of this code is given by
$R=(\log M)/n$.
We use
the maximum likelihood decoding
\begin{align}
\hat{\bm{X}}=\argmax_{j\in\{1,2,\cdots,M\}}\sum_{i=1}^n \log W(Y_i|(\bm{X}_j)_i)\per\n
\end{align}
We mainly consider the case that ties are broken uniformly at random.
See Sect.~\ref{sec_tie} for
ties immediately regarded as a decoding error.
Note that the former case corresponds to \cite{dobrushin}
and the latter case is considered in \cite{gallager_tight}.

For a random variable $V$ we write
$\bar{V}$ to denote the empirical mean of $n$ i.i.d.~copies
and write $\tilde{V}=\sqn (\bar{V}-\E[\bar{V}])$.
We write $a\land b=\min\{a,b\}$ and $a\lor b=\max\{a,b\}$.
For $x=0$ we define $(\e^x-1)/x=x/(\e^x-1)=1$.

\subsection{Error Exponent}
Define a random variable $Z(\la)$ on the space of functions
$\mathbb{R}\to\mathbb{R}$ by
\begin{align}
Z(\la)&=
\log\E_{X'}\left[\e^{\la \nu(X,Y,X')}\right]\n
\end{align}
and its derivatives by
\begin{align}
\bibunZ{m}(\la)&=
\frac{\rd^m}{\rd \la^m}\log\E_{X'}\left[\e^{\la \nu(X,Y,X')}\right]\com\n
\end{align}
which we also write as $Z'(\la),\,Z''(\la),\cdots$.
Here $\E_{X'}$ denotes the expectation over $X'$
for given $(X,Y)$.
We define
\begin{align}
Z(\la+\im \xi)&=
\log
\E_{X'}\left[\e^{(\la +\im \xi)\nu(X,Y,X')}\right]
\nn
\cumu(\la+\im \xi)&=
\log
\left|
\E_{X'}\left[\e^{(\la +\im \xi)\nu(X,Y,X')}\right]
\right|
\com\n
\end{align}
where $\la,\xi\in \mathbb{R}$ and $\im$ is the imaginary unit.
Here we always consider the case $\la>0$ and define $\e^{(\la+\im \xi)(-\infty)}=0$.

The random coding error exponent
for $0<R<I(X;Y)$
is
denoted by
\begin{align}
E_r(R)
&=
-\inf_{(\alpha,\la)\in [0,1]\times [0,\infty)}\{\alpha R+\log \E[\e^{\alpha Z(\la)}]\}\nn
&=
-\min_{\alpha\in (0,1]} \{\alpha R+\log \E[\e^{\alpha Z(1/(1+\alpha))}]\}\com
\label{def_exponent}
\end{align}
and we write the optimal solution of $(\alpha,\la)$ as
$(\rho,\eta)=(\rho,1/(1+\rho))$.
Critical rate $R_{\mathrm{crit}}$
is the largest $R$ such that the optimal solution of \eqref{def_exponent}
is $\rho=1$.

In the strict sense
the random coding error exponent represents
the supremum of \eqref{def_exponent} over $P_X$
but for simplicity
we fix $P_X$ and omit its dependence.
See \cite[Theorem 2]{altug_journal} for
a condition that there exists $P_X$ which attains
this supremum.

Let $P_{\rho}$ be the probability measure
such that
$\rd P_{\rho}/\rd P=\e^{\rho Z(\eta)-\La(\rho)}$
for
$\La(\rho)=\log \E[\e^{\rho Z(1/(1+\rho))}]$.
We write the expectation under $P_{\rho}$ by
$\E_{\rho}$
and define
\begin{align}
\mu_i&=\E_{\rho}[\bibunZ{i}(\eta)]
=\e^{-\La(\rho)}\E[\bibunZ{i}(\eta)\e^{\rho Z(\eta)}]\nn
\si_{ij}&=\E_{\rho}[(\bibunZ{i}(\eta)-\mu_i)(\bibunZ{j}(\eta)-\mu_j)]\nn
&=\e^{-\La(\rho)}\E[(\bibunZ{i}(\eta)-\mu_i)(\bibunZ{j}(\eta)-\mu_j)\e^{\rho Z(\eta)}]\nn
\Sigma&=\left(
\begin{array}{cc}
\si_{00}&\si_{01}\\
\si_{10}&\si_{11}
\end{array}
\right)\per
\n
\end{align}
By letting $\Delta=-(\mu_0+R)$
we have
$\Delta>0$ if $R<\Rcrit$
and $\Delta=0$ otherwise.
For a one-dimensional random variable $V\in \uR$, we say that
$V$ is singular if $V\in \{-\infty,v\}$ a.s.~for some $v\in \bbR$.

\begin{definition}\label{def_lattice}{\rm
Channel $W$ is singular if
$\nu(X,Y,X')$ given $(X,Y)$ is singular almost surely, that is,
$P_{X'}[\nu(X,Y,X')\in \{-\infty,0\}]=1$ a.s.
}\end{definition}
As discussed in \cite{dobrushin},
$\mu_2=0$ if $W$ is singular and $\mu_2>0$ otherwise.


\subsection{Lattice and Nonlattice Distributions}
We call that nonsingular one-dimensional random variable $V\in\uR$ has a lattice distribution
with span $h>0$ and offset $a\in\bbR$
if $V\in \{a+ih: i\in\bbZ\}\cup\{-\infty\}$ a.s.~and
$h$ is the largest one satisfying this property.

Let $a\in\bbR^2$ be arbitrary
and $h^{(1)},h^{(2)}\in \bbR^2$ be linearly independent vectors.
We say that
two-dimensional random variable $V\in \bbR^2$ with covariance matrix $\Sigma$
satisfying $|\Sigma|\neq 0$
has a lattice distribution over $L=\{a+i h^{(1)}+j h^{(2)}: i,j\in\bbZ\}$
if $V \in L$ a.s.~and
no sublattice of $L$ satisfies this property.
We say that $V\in \bbR^2$
has a lattice-nonlattice distribution over
set $L'=\{a+i h^{(1)}+t h^{(2)}: i\in\bbZ, t\in\bbR\}$
of lines with equal interval
if $V \in L'$ a.s.~and $(h^{(1)},h^{(2)})$ is a pair with largest
$|\det(h^{(1)},h^{(2)})|/\Vert h^{(2)}\Vert$.
We say that $V\in \bbR^2$ has a strongly nonlattice distribution
if $V$ does not have a lattice distribution or lattice-nonlattice distribution.

\begin{definition}\label{def_lattice}{\rm
Channel $W$ is $h$-lattice
if
$\nu(X,Y,X')$ has a lattice distribution with span $h$
and is nonlattice otherwise.
We define the span of a nonlattice channel
as $h=0$.
}\end{definition}
Note that if $W$ is $h$-lattice then
the offset of $\nu(X,Y,X')$ is zero
from the definition of $\nu$.
Whereas this classification of a channel
also appears in many studies such as \cite{gallager_tight},
we also consider another classification to derive a tight bound.
This classification also depends on $\eta=1/(1+\rho)$
that is determined from $R$.
\begin{definition}\label{def_lattice}{\rm
Channel and rate pair $(W, R)$
is $(h',a')$-lattice
if $Z(\eta)$ has a lattice distribution with span $h'$
and offset $a'$,
and is nonlattice otherwise.
The pair $(W,R)$ is pseudo-symmetric if
$(Z(\eta),Z'(\eta))$ is distributed over some single line,
that is, $Z'(\eta)$ is a linear function of $Z(\eta)$.
}\end{definition}

Dobrushin \cite{dobrushin} considered
the case that
$W$ is a symmetric discrete channel in the strong sense that
each row and column of the transition probability matrix
are permutations of the others.
In this case
the conditional distribution of $W(Y|X')$ given $Y$
does not depend on $Y$
and therefore for any $y_0\in\calY$ we have
\begin{align}
Z(\eta)&=\log \E_{X'}[W(y_0|X')^{\eta}]-\eta \log W(Y|X)\com\nn
Z'(\eta)&=\frac{\E_{X'}[W(y_0|X')^{\eta}\log W(y_0|X')]}{\E_{X'}[W(y_0|X')^{\eta}]}
-\log W(Y|X)\per\n
\end{align}
The first terms of RHSs of them are constants
and the following property trivially holds.
\begin{proposition}
Assume that discrete channel $W$ is strongly symmetric.
Then $(W,R)$ is pseudo-symmetric for any $R$.
Furthermore, $W$ is $h$-lattice if and only if
$(W,R)$ is $(\eta h, a)$-lattice for some $a\in\bbR$.
\end{proposition}
We can see from this proposition that
symmetric channels considered in \cite{dobrushin}
correspond to the degenerate case
where $(Z(\eta),Z'(\eta))$ is linearly dependent.


As in \cite{exact_isit} we always assume that
for lattice span $h\ge 0$ of $W$
there exist $\alpha,\gapxia>0$ and a neighborhood $\mathcal{S}\ni \la$ of $\eta$
such that for any $0<\gapxi<b_2<2\pi/h\le \infty$
\begin{align}
&\sup_{\la\in \mathcal{S}}\E_{\rho}[\e^{\alpha |\bibunZ{i}(\la)|}]<\infty\com\quad i=1,2,3,\nn
&\sup_{\la\in \mathcal{S},\,\xi\in[-\gapxia,\gapxia]}
\E_{\rho}[\e^{\alpha |(\partial^4/\partial \xi^4)Z(\la+\im\xi)|}]<\infty\com\nn
&\sup_{\la\in \mathcal{S},\,\xi\in [\gapxi,b_2]}
\E_{\rho}[\e^{\alpha |\cumu(\la+\im\xi)-\cumu(\la)|}]<\infty\com\n
\end{align}
which are trivially satisfied for finite discrete channels.


\section{Exact Asymptotics for Nonsingular Channels}
In this section we derive the exact asymptotics for
nonsingular channels covering results
in \cite{dobrushin}\cite{exact_isit} as special cases.
First we give the exact asymptotics for
$R\le \Rcrit$.
\begin{theorem}\label{thm_equal}
Let $W$ be a channel with lattice span $h\ge 0$ of $W$.
Then
\begin{align}
P_{\mathrm{RC}}(n)
=
\begin{cases}
\frac{(1+\so(1))h(\e^{h/2}+1)}{2(\e^{h/2}-1)\sqrt{2\pi n(\mu_2+\si_{11})}}
\e^{-nE_r(R)},& \mbox{if $R<\Rcrit$},\\
\frac{(1+\so(1))h(\e^{h/2}+1)}{4(\e^{h/2}-1)\sqrt{2\pi n(\mu_2+\si_{11})}}
\e^{-nE_r(R)},& \mbox{if $R=\Rcrit$}.
\end{cases}
\label{bound_equal}
\end{align}
\end{theorem}
We prove this theorem in Appendix \ref{sect_equal}
using
two-dimensional Berry-Esseen bound (or one-dimensional one for
pseudo-symmetric $(W,R)$) in \cite{berry_multi}.

The derived bound is equal to those of \cite{gallager_tight} (for $R<\Rcrit$)
and \cite{dobrushin} (for strongly symmetric channels)
when $W$ is nonlattice, whereas these three bounds are different to each other
for the lattice case.
Gallager \cite{gallager_tight} derived a bound
for ties regarded as errors
and the bound in this theorem for uniformly broken ties
is slightly smaller than the bound in \cite{gallager_tight}
as discussed in Sect.~\ref{sec_tie}.
On the other hand, Dobrushin \cite{dobrushin} considered
uniformly broken ties but the explicit expression
on the constant factor was not derived for this case.

Now we consider the case $R>R_{\mathrm{crit}}$.
In this case
the bound also depends on
whether $(W,R)$ is lattice or not
and becomes much complicated.
For $h\ge 0$, let
\begin{align}
g_h(u)&=
1-\frac{\e^{-\frac{h\eta}{\e^{h\eta}-1}u}(1-\e^{-h\eta u})}{h\eta u}\nn
g_{\rho,h}(u)&=u^{-\rho}g_h(u)
=
\frac{1}{u^{\rho}}-\frac{\e^{-\frac{h\eta}{\e^{h\eta}-1}u}(1-\e^{-h\eta u})}{h\eta u^{1+\rho}}
\nn
\psi_{\rho,h,h'}(x)&=\sum_{i\in\bbZ}h'g_{\rho,h}(\e^{x+ih'})\nn
\psi_{\rho,h}&=
\int_{w\in\bbR}g_{\rho,h}(\e^{w})\rd w\nn 
&=
\frac{\Gamma(1-\rho)}{\rho}
\left(\frac{h\eta}{\e^{h\eta}-1}\right)^{\rho+1}
\frac{\e^{h}-1}{h}
\com\label{defs_constants}
\end{align}
where $\Gamma(\cdot)$ is Gamma function.
Note that $\psi_{\rho,h,h'}(x)$ is a periodic function with period $h'$ and
satisfies $\psi_{\rho,h}=\lim_{h'\downarrow 0}\psi_{\rho,h,h'}(x)$ for any $x\in\bbR$.
The following theorem is the main contribution of this paper,
which solves the exact asymptotics of random coding error probability
for rate above the critical rate.
\begin{theorem}\label{thm_main}
Fix $R\in (R_{\mathrm{crit}},I(X;Y))$
and let $h>0$ be the lattice span of channel $W$.
Then
\begin{align}
P_{\mathrm{RC}}(n)
=
(1+\so(1))\frac{(1+\rho)^{\rho}I_n}{\sqrt{(2\pi)^{1+\rho}\mu_2^{\rho}}}
 n^{-(1+\rho)/2}\e^{-nE(R)}
\com\n
\end{align}
where, if $(W,R)$ is nonlattice then
\begin{align}
I_n
&=
I=\frac{
\psi_{\rho,h}
}{\sqrt{\sigma_{00}+\rho |\Sigma|/\mu_2}}\label{bound_n_n}
\end{align} 
and if $(W,R)$ is $(h',a')$-lattice then
\begin{align}
I_n
&=
\frac{\E_{V}
\sqx{\psi_{\rho,h,h'}\left(na'-\frac{|\Sigma|V^2}{2(\sigma_{00}+\rho|\Sigma|/\mu_2)}-\log c_2\sqn\right)}
}{\sqrt{\sigma_{00}+\rho |\Sigma|/\mu_2}}\label{bound_n_l}
\end{align}
for standard normal $V$.
In particular, if $(W,R)$ is pseudo-symmetric then
\begin{align}
\!\!\!\!
I_n
=
\begin{cases}
\frac{
\psi_{\rho,h}
}{\sqrt{\sigma_{00}}},
&\mbox{$(W,R)$ is nonlattice,}\\
\frac{
\psi_{\rho,h,h'}\left(na'-\log c_2\sqn\right)
}{\sqrt{\sigma_{00}}},&\mbox{$(W,R)$ is $(h',a')$-lattice.}\\
\end{cases}
\label{bound_s}
\end{align}%
\end{theorem}
We give a sketch of a proof in Sect.~\ref{sec_sketch}
and the full proof is in Appendix \ref{append_main}.
If $(W,R)$ is pseudo-symmetric then $|\Sigma|=0$
and 
the bound
\eqref{bound_s} is
a special case of \eqref{bound_n_n} and \eqref{bound_n_l},
although
the proof is given separately.

The bound in \cite{exact_isit} for strongly nonlattice $(Z(\eta),Z'(\eta))$
is a special case of \eqref{bound_n_n}.
If $(W,R)$ is lattice then
$\PRC(n)\cdot \allowbreak n^{(1+\rho)/2}\e^{nE(R)}$
does not converge as shown in this theorem.
This phenomenon is suggested in \cite{dobrushin}
for strongly symmetric channel, which is a special case of pseudo-symmetric $(W,R)$.
Furthermore,
if $(W,R)$ is not pseudo-symmetric
then
$\PRC(n)n^{(1+\rho)/2}\e^{nE(R)}$ is expressed
as an expectation of a periodic function for a normal random variable,
which seems to be impossible to integrate out analytically.
The known bounds are summarized in Fig.~\ref{fig_summary}
and the derived bound in this paper covers all region.

\begin{figure}[t]%
  \begin{center}
   \includegraphics[bb=30 80 380 800,angle=270,clip, width=60mm]{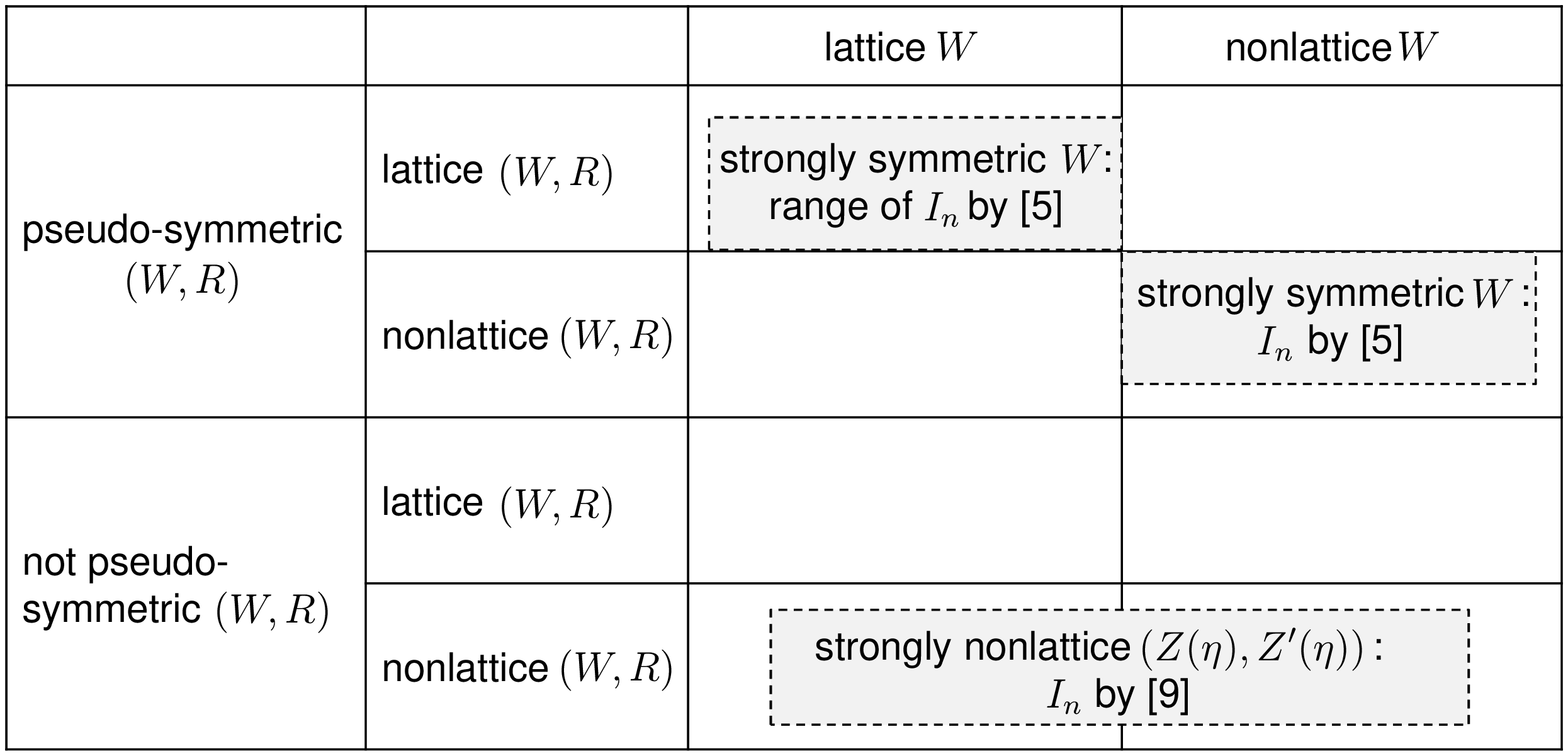}
  \end{center}
  \caption{Known bounds that are proved to be tight up to constant factors for $\Rcrit<R<I(X;Y)$.}
  \label{fig_summary}
\end{figure}%

\section{Exact Asymptotics for Singular Channels}
Now we consider the singular channels,
which satisfies $\mu_2=0$, that is,
$P_{X'}[\nu(X,Y,X')\in \{0,1\}]$ a.s.~for $(X,Y)$.
Proofs of theorems are given in Appendix \ref{sect_proof_sin}.

As in the case of nonsingular channels,
we have a simple expression
of the error probability for $R\le \Rcrit$.
\begin{theorem}\label{thm_equal_sin}
If channel $W$ is singular and has lattice span $h\ge 0$ then,
for $R\le \Rcrit$,
\begin{align}
P_{\mathrm{RC}}(n)
=
\begin{cases}
(1/2+\so(1))
\e^{-nE_r(R)},& \mbox{if $R<\Rcrit$},\\
(1/4+\so(1))
\e^{-nE_r(R)},& \mbox{if $R=\Rcrit$}.
\end{cases}
\label{bound_sin_equal}
\end{align}
\end{theorem}
The bound in \cite{dobrushin} is a special case
of this bound for strongly symmetric channels.
As pointed out in \cite{altug},
the bound derived in \cite{gallager_tight}
does not apply
for the case of nonsingular channels.
Whereas \cite{altug} derives a range of
$P_{\mathrm{RC}}(n)/\e^{-nE_r(R)}$,
this theorem derives its exact value for $n\to\infty$.

Now we consider the case
$R>\Rcrit$.
Dobrushin \cite{dobrushin} pointed out that
$\Rcrit=I(X;Y)$ holds when
a strongly symmetric channel is singular,
which means that $\Rcrit\le R<I(X;Y)$ never occurs in this case.
For general cases, we can see
from
the definition of $\Rcrit$ given below \eqref{def_exponent}
that
$\Rcrit=I(X;Y)$ if and only if
$Z(\eta)$ is singular, that is, $Z(\eta)$ is a constant random variable.
Thus, we can always assume that
$Z(\eta)$ is not singular when
$\Rcrit< R<I(X;Y)$.

The exact asymptotics for
this rate region is given based on the following values.
\begin{align}
\gsin(u)&=
1-\frac{1-\e^{-u}}{u}\com\nn
\gsin_{\rho}(u)&=u^{-\rho}\gsin_h(u)
=
\frac{1}{u^{\rho}}-\frac{1-\e^{-u}}{u^{1+\rho}}\com
\nn
\psin_{\rho,h'}(x)&=\sum_{i\in\bbZ}h'\gsin_{\rho}(\e^{x+ih'})\com\nn
\psin_{\rho}&=
\int_{w\in\bbR}\gsin_{\rho}(\e^{w})\rd w
=
\frac{\Gamma(1-\rho)}{\rho(1+\rho)}
\per\n
\end{align}
\begin{theorem}\label{thm_above_sin}
Assume that channel $W$ is singular.
Then, for $\Rcrit<R<I(X;Y)$,
\begin{align}
\lefteqn{
P_{\mathrm{RC}}(n)
}\nn
&\!=
\begin{cases}
\frac{(1+\so(1))\psin_{\rho,h'}(na')}{\sqrt{2\pi n}\sigma_{00}}
\e^{-nE_r(R)},&\!\!\mbox{if $(W,R)$ is $(h',a')$-lattice,}\\
\frac{(1+\so(1))\psin_{\rho}}{\sqrt{2\pi n}\sigma_{00}}
\e^{-nE_r(R)},&\!\!\mbox{if $(W,R)$ is nonlattice.}
\end{cases}\n
\end{align}
\end{theorem}

\newcommand{\Vz}{V_0}
\newcommand{\Vo}{V_1}
\newcommand{\vz}{v_0}
\newcommand{\vo}{v_1}
\newcommand{\tV}{\tilde{V}}
\newcommand{\hzo}{h_0^{(1)}}
\newcommand{\hoo}{h_1^{(1)}}
\newcommand{\hzt}{h_0^{(2)}}
\newcommand{\hot}{h_1^{(2)}}



%

\newcommand{\bZz}{\bar{Z}^{(0)}}
\newcommand{\bZo}{\bar{Z}^{(1)}}
\newcommand{\tZz}{\tilde{Z}^{(0)}}
\newcommand{\tZo}{\tilde{Z}^{(1)}}

\section{Bounds for Ties Regarded as Errors}\label{sec_tie}
In this section we discuss how the bound changes
when a tie of likelihoods is immediately regarded as
a decoding error.

First we consider nonsingular channels.
Let $p_0=p_0(\bmx,\bmy)$ and $p_+=p_+(\bmx,\bmy)$ be
probabilities that the likelihood of a codeword $\bmX'$
equals
and exceeds that of the sent sequence $\bmx$ given received sequence $\bmy$, respectively.
Then the error probability over $M$ codewords
is
\begin{align}
\lefteqn{
q_M(\pp,\pz)
}\nn
&=
1-(1-\pp)^{M-1}
\nn
&
\quad+
\sum_{i=1}^{M-1}\pz^i(1-\pp-\pz)^{M-i-1}{{M-1}\choose i}\left(1-\frac{1}{i+1}\right)\nn
&=
1-\frac{(1-\pp)^{M}-(1-\pz-\pp)^{M}}{M\pz}
\label{error_uniform}
\end{align}
for ties broken uniformly at random and
\begin{align}
\tilde{q}_M(\pp,\pz)&=
1-(1-p_0-p_+)^{M-1}
\n
\end{align}
for ties regarded as errors.
By following the analysis for \eqref{error_uniform}
in \cite{exact_isit}
we can see that $g_h(u)$ in Prop.~\ref{expansion_first}
is replaced with
\begin{align}
\tilde{g}_h(u)=1-\e^{-\frac{h\eta\e^{h\eta}}{\e^{h\eta}-1}u}\per\n
\end{align}
In the case of $R\le \Rcrit$,
the value $\lim_{u\downarrow0}g_h(u)$ only affects the analysis
and the bound becomes
\begin{align}
\frac{\lim_{u\downarrow 0}\tilde{g}_h(u)}{\lim_{u\downarrow 0}g_h(u)}
=
\frac{2\e^{h/2}}{\e^{h/2}+1}\in [1,2)
\n
\end{align}
times that in Theorem \ref{thm_equal}, that is,
\eqref{bound_equal} is replaced with
\begin{align}
P_{\mathrm{RC}}(n)
=
\begin{cases}
\frac{(1+\so(1))h\e^{h/2}}{(\e^{h/2}-1)\sqrt{2\pi n(\mu_2+\si_{11})}}
\e^{-nE_r(R)},& \!\!\mbox{if $R<\Rcrit$},\\
\frac{(1+\so(1))h\e^{h/2}}{2(\e^{h/2}-1)\sqrt{2\pi n(\mu_2+\si_{11})}}
\e^{-nE_r(R)},& \!\!\mbox{if $R=\Rcrit$},
\end{cases}\n
\end{align}
which reproduces the bound in \cite{gallager_tight}
for $R<\Rcrit$.
For the case $\Rcrit<R<I(X;Y)$,
values
$\psi_{\rho,h,h'}(x)$
and $\psi_{\rho,h}$ in \eqref{defs_constants}
change accordingly to the change of the function $g_{h}(u)$ to $\tilde{g}_h(u)$.
In particular, we can see that
$\psi_{\rho,h}$ is replaced with
\begin{align}
\tilde{\psi}_{\rho,h}&=
\int_{w\in\bbR}\e^{-\rho w}\tilde{g}_{h}(\e^{w})\rd w
=
\Gamma(1-\rho)\left(\frac{h\eta\e^{h/2}}{\e^{h\eta}-1}\right)^{\rho}\com\n
\end{align}
which satisfies
\begin{align}
\frac{\tilde{\psi}_{\rho,h}}{\psi_{\rho,h}}
&=
(1+\rho)
\frac{\e^{h}-\e^{\rho h/(1+\rho)}}{\e^h-1}
\in [1,1+\rho)\subset [1,2)\per\n
\end{align}

Next we consider singular channels.
In this case, the decoding error probability
for uniformly broken ties, which will be given in \eqref{error_sin_uniform}
of Appendix \ref{sect_proof_sin},
changes to
\begin{align}
\tilde{q}_M(\pz)&=
1-(1-p_0)^{M-1}.\n
\end{align}
We can adapt the proofs of Theorems \ref{thm_equal_sin} and \ref{thm_above_sin}
to this change
by simply replacing $\gsin(u)=1-(1-\e^{-u})/u$
with $\tgsin(u)=1-\e^{-u}$.
By this replacement
the bound \eqref{bound_sin_equal} becomes doubled,
that is, we have
\begin{align}
P_{\mathrm{RC}}(n)
=
\begin{cases}
(1+\so(1))
\e^{-nE_r(R)},& \mbox{if $R<\Rcrit$},\\
(1/2+\so(1))
\e^{-nE_r(R)},& \mbox{if $R=\Rcrit$},
\end{cases}\n
\end{align}
since we have
$\lim_{u\to 0}\tgsin(u)/\lim_{u\to 0}\gsin(u)=2$.

For the case $\Rcrit<R<I(X;Y)$,
values
$\psin_{\rho,h'}(x)$
and $\psin_{\rho}$ in \eqref{defs_constants}
change accordingly to the change of the function $\gsin(u)$ to $\tgsin(u)$.
In particular,
$\psin_{\rho}$ is replaced with
\begin{align}
\tpsin_{\rho}&=
\int_{w\in\bbR}\e^{-\rho w}\tgsin(\e^{w})\rd w
=
\frac{\Gamma(1-\rho)}{\rho}\com\n
\end{align}
which satisfies
$\tpsin_{\rho}/\psin_{\rho}
=(1+\rho)
\in [1,2)$.

\section{Proof Outline of Theorem \ref{thm_main}}\label{sec_sketch}

In this section we give a rough derivation of
\eqref{bound_n_l} in Theorem \ref{thm_main},
which is the most difficult part of the results in this paper.
See Appendix \ref{append_main} for the full proof.
We start with the following fact
derived in \cite{exact_isit}.
\begin{proposition}[{\cite[Theorem 1]{exact_isit}}]\label{expansion_first}
For
lattice span $h\ge 0$ of channel W,
arbitrary $\epsilon_1>0$ and sufficiently small $\epsilon_2>0$,
there exists $n_0>0$ such that for all $n\ge n_0$
\begin{align}
\lefteqn{
(1-\ep_1)
\E\!\left[
g_h\left(
(1-\ep_1)\frac{
 \e^{n(\bar{Z}(\eta)+R-(\bar{Z}'(\eta))^2/2(\mu_2-\ep_2))}
}{\eta\sqrt{2\pi n \mu_2}}
\right)
\right]
}\nn
&\le
P_{\mathrm{RC}}(n)\nn
&\le
(1+\ep_1)
\E\!\left[
g_h\left(
(1+\ep_1)\frac{
 \e^{n(\bar{Z}(\eta)+R-(\bar{Z}'(\eta))^2/2(\mu_2+\ep_2))}
}{\eta\sqrt{2\pi n \mu_2}}
\right)
\right]\per
\n
\end{align}
\end{proposition}

In the following we write $(Z_0,Z_1)$ instead of
$(Z(\eta),Z'(\eta))$ for notational simplicity.
For its empirical mean $(\bZ_0,\bZ_1)$,
we write $(\tZ_0,\tZ_1)=\sqn(\bZ_0-\mu_0,\bZ_1-\mu_1)
=\sqn(\bZ_0+R+\Delta,\bZ_1)$.

We can show Theorem \ref{thm_main} (and Theorem \ref{thm_equal})
by evaluating
\begin{align}
\lefteqn{
\E\!\left[
g_h\left(
\frac{
 \e^{n(\bZ_0+R-\bZ_1^2/2c_1)}
}{c_2\sqrt{n}}
\right)
\right]
}\nn
&=
\e^{-nE_r(R)}
\E_{\rho}\left[
\e^{-n\rho \bZ_0}
g_h\left(
\frac{
 \e^{n(\tZ_0+R-\bZ_1^2/2c_1)}
}{c_2\sqrt{n}}
\right)
\right]\nn
&=
\frac{\e^{-nE_r(R)}}{(c_2\sqn)^{\rho}}
\E_{\rho}\left[
\e^{-\rho \tZ_1^2/2c_1}
g_{\rho,h}\left(
\frac{
 \e^{\sqrt{n}\tZ_0-n\Delta-\tZ_1^2/2c_1}
}{c_2\sqrt{n}}
\right)
\right]\label{tilting}
\end{align}
for fixed $c_1,c_2>0$
and letting $c_1:=\mu_2\pm\ep_2,\,c_2:=\eta\sqrt{2\pi \mu_2}/(1\pm\epsilon_1)$
and finally letting $\ep_1,\ep_2\downarrow 0$.

For evaluation of the expectation in \eqref{tilting}
we use a version of bivariate local limit theorem,
which is obtained by ``blurring'' a standard bivariate local limit theorem
in one direction.
Let $\phi_{\Sigma}$ be the density function of normal distribution
with zero mean and covariance matrix $\Sigma$.
Then the following lemma holds for
random variable $V=(\Vz,\Vo)\in\bbR^2$
with zero mean and covariance matrix $\Sigma$
such that $|\Sigma|>0$.
\begin{lemma}\label{lem_blur}
Fix $\delta>0$ and a sequence $b_n>0$ such that $b_n=\so(\sqn)$
and $\lim_{n\to\infty}\allowbreak b_n=\infty$.
If $\Vz$ has a lattice distribution with span $h$ and offset $a$ then
\begin{align}
\frac{n}{hb_n}\Pr[\sqrt{n}(\tV-v)\in \{0\} \times [0,b_n)]\to \phi_{\Sigma}(v)
\n
\end{align}
as $n\to\infty$ uniformly for
$v\in \{\sqn a + ih/\sqn: i\in\bbZ\}\times \bbR$.
If $\Vz$ does not have a lattice distribution then
\begin{align}
\frac{n}{\delta b_n}\Pr[\sqrt{n}(\tV-v)\in [0,\delta) \times [0,b_n)]\to \phi_{\Sigma}(v)
\n
\end{align}
as $n\to\infty$ uniformly for $v\in \bbR^2$.
\end{lemma}
This lemma evaluates
the distribution of $\tV=(\tV_0,\tV_1)$
with $\lo(1/\sqn)$ precision in $\tV_0$ direction
and with $\lo(b_n/\sqn)$ precision in $\tV_1$ direction.

\begin{proof}[Proof Sketch of Theorem \ref{thm_main}]
Let $\calZ_{n}=\{(na'+ih')/\sqn: i\in\bbZ\}$.
Then we obtain from the bivariate local limit theorem in
Lemma \ref{lem_blur} that
\begin{align}
\lefteqn{
\E_{\rho}\left[
\e^{-\rho \tZ_1^2/2c_1}
g_{\rho,h}\left(
\frac{
 \e^{\sqrt{n}\tZ_0-\tZ_1^2/2c_1}
}{c_2\sqrt{n}}
\right)
\right]
}\nn
&\approx\!
\sum_{z_0 \in \calZ_{n}}
\!\int_{z_1}\!
\frac{h'\phi_{\Sigma}(z_0,z_1)}{\sqn}
\e^{-\rho z_1^2/2c_1}
g_{\rho,h}\left(
\frac{
 \e^{\sqrt{n}z_0-z_1^2/2c_1}
}{c_2\sqrt{n}}
\right)\!
\rd z_1.\n
\end{align}
Here it holds
from $g_h(u)\le (1+h\eta)(1\land u)$ \cite[Lemma 8]{exact_isit} that
\begin{align}
g_{\rho,h}(\e^w)\le (1+h\eta)(\e^{-\rho w}\land \e^{(1-\rho)w})\com\label{grh}
\end{align}
which
means that
$g_{\rho,h}\left(
\frac{
 \e^{\sqrt{n}z_0-z_1^2/2c_1}
}{c_2\sqrt{n}}
\right)$ decays exponentially for $|z_0|=\Omega(1/\sqn)$ (as far as $z_1=\lo(1)$ holds).
Thus
\begin{align}
\lefteqn{
\sum_{z_0 \in \calZ_{n}}
\int_{z_1}\!
\frac{h'\phi_{\Sigma}(z_0,z_1)}{\sqn}
\e^{-\rho z_1^2/2c_1}
g_{\rho,h}\left(
\frac{
 \e^{\sqrt{n}z_0-z_1^2/2c_1}
}{c_2\sqrt{n}}
\right)
\rd z_1
}\nn
&\approx
\sum_{z_0 \in \calZ_{n}}
\int_{z_1}\!
\frac{h'\phi_{\Sigma}(0,z_1)}{\sqn}
\e^{-\rho z_1^2/2c_1}
g_{\rho,h}\left(
\frac{
 \e^{\sqrt{n}z_0-z_1^2/2c_1}
}{c_2\sqrt{n}}
\right)
\rd z_1\nn
&=
\!\int_{z_1}\!
\frac{\phi_{\Sigma}(0,z_1)}{\sqn}
\e^{-\rho z_1^2/2c_1}
\psi_{\rho,h,h'}(na-z_1^2/2c_1-\log c_2\sqn)
\rd z_1\nn
&=
\int_{z_1}
\frac{\psi(na'-z_1^2/2c_1-\log c_2\sqn)}{2\pi \sqrt{n|\Sigma|}}
\e^{-\pax{\frac{\rho}{c_1}+\frac{\si_{00}}{|\Sigma|}}\frac{z_1^2}{2}}\rd z_1\nn
&=
\frac{
\E_{V}\sqx{
\psi\pax{na'-\frac{|\Sigma|V^2}{2(\sigma_{00}+\rho |\Sigma|/c_1)}-\log c_2\sqn}
}}{\sqrt{2\pi n(\sigma_{00}+\rho|\Sigma|/c_1)}}\per\n
\end{align}
We obtain \eqref{bound_n_l} by combining
Prop.~\ref{expansion_first} with \eqref{tilting}
by letting $c_1:=\mu_2$ and $c_2:=\eta\sqrt{2\pi\mu_2}$.
\end{proof}

\section*{Acknowledgment}
The author thanks Dr.~Junpei Komiyama for discussion on the equidistribution theorem.



\appendices
\section{Bivariate Local Limit Theorem with Anisotropic Resolution}\label{sec_decomp}
In this section we show a version of bivariate local limit theorem
suitable for the proof of Theorem \ref{thm_main}
by ``blurring'' a standard bivariate local limit theorem
in one direction.
Let $\phi_{\Sigma}$ be the density function of normal distribution
with zero mean and covariance matrix $\Sigma$.
The goal of this section is to prove the following lemma for
random variable $V=(\Vz,\Vo)\in\bbR^2$
with zero mean and covariance matrix $\Sigma$
such that $|\Sigma|>0$.
\begin{lemmat}[restated]
Fix $\delta>0$ and a sequence $b_n>0$ such that $b_n=\so(\sqn)$
and $\lim_{n\to\infty}\allowbreak b_n=\infty$.
If $\Vz$ has a lattice distribution with span $h$ and offset $a$ then
\begin{align}
\frac{n}{hb_n}\Pr[\sqrt{n}(\tV-v)\in \{0\} \times [0,b_n)]\to \phi_{\Sigma}(v)\label{lem_lattice}
\end{align}
as $n\to\infty$ uniformly for
$v\in \{\sqn a + ih/\sqn: i\in\bbZ\}\times \bbR$.
If $\Vz$ does not have a lattice distribution then
\begin{align}
\frac{n}{\delta b_n}\Pr[\sqrt{n}(\tV-v)\in [0,\delta) \times [0,b_n)]\to \phi_{\Sigma}(v)
\label{lem_nonlattice}
\end{align}
as $n\to\infty$ uniformly for $v\in \bbR^2$.
\end{lemmat}

We show this lemma
based on
Prop.~\ref{bivariate} given below.%

\begin{proposition}[{Bivariate Local Limit Theorem\footnote{%
Adapted from the original version in \cite{local_bivariate}
for the case that $\Sigma$ is the identify matrix
and $h^{(1)}$ and $h^{(2)}$ are unit vectors.}
\cite[Theorems 1--3]{local_bivariate}}]\label{bivariate}
Let $h^{(1)},h^{(2)}\in \bbR^2$ be linearly independent vectors.
If $V$ has a strongly nonlattice distribution then
\begin{align}
&
n\Pr[\sqn (\tV-v)\in [0,\delta_0)\times [0,\delta_1)]
\to \delta_0\delta_1\phi_{\Sigma}(v)\label{prop_nonlattice}
\end{align}
as $n\to\infty$ uniformly for $v\in \bbR^2$
and $\delta_0,\delta_1$ in a compact subset of $(0,\infty)$.
If $V$ has a lattice-nonlattice distribution over $L'=\{a+i h^{(1)}+th^{(2)}:i\in\bbZ,\,t\in\bbR\}$ then
\begin{align}
\frac{n}{H}\Pr[\sqn(\tV-v)\in \{th^{(2)}: t\in [0,\delta)\}]-\delta\phi_{\Sigma}(v)
\!\to\! 0\label{ln_local}
\end{align}
as $n\to\infty$ uniformly for $v\in \{(na+i h^{(1)}+th^{(2)})/\sqrt{n}:i\in\bbZ,\,t\in\bbR\}$
and $\delta$ in a compact subset of $(0,\infty)$,
where $H=|\det(h^{(1)},h^{(2)})|$.
If $V$ has a lattice distribution over $L=\{a+i h^{(1)}+jh^{(2)}:i,j\in\bbZ\}$ then
\begin{align}
\frac{n}{H}\Pr[\tV=v]-\phi_{\Sigma}(v)\to 0\label{lattice_local}
\end{align}
as $n\to\infty$ uniformly for $v\in \{(na+i h^{(1)}+jh^{(2)})/\sqrt{n}:i,j\in\bbZ\}$.
\end{proposition}
\begin{proof}[Proof of Lemma \ref{lem_blur}]
If $V$ has a strongly nonlattice distribution
then \eqref{lem_nonlattice} is straightforward from
\eqref{prop_nonlattice}
and we consider the other case that
$V$ has a lattice distribution
on $L=\{a+ih^{(1)}+jh^{(2)}: i,j\in\bbZ\}$
or a lattice-nonlattice distribution
on $L=\{a+ih^{(1)}+th^{(2)}: i\in\bbZ,t\in\bbR\}$.
We define
$L_n=\{na+ih^{(1)}+jh^{(2)}: i,j\in\bbZ\}$
and $L_n=\{na+ih^{(1)}+th^{(2)}: i\in\bbZ,t\in\bbR\}$
for these cases, where we assume without loss of generality
that $\Vert h^{(2)}\Vert=1$ for the latter case.

Define
$\mathrm{Vol}(S)$ as the total lengths of
lines if $S$ is a subset of parallel lines
and as the total number of points in $S$
if $S$ is a subset of a lattice.
Then, it suffices to show from \eqref{ln_local} and \eqref{lattice_local} that
\begin{align}
\frac{H}{\delta b_n}\mathrm{Vol}((L_n-v)\cap [0,\delta)\times [0,b_n))\to 1\label{vol1}
\end{align}
as $n\to\infty$ uniformly for $v\in\bbR$
if $\Vz$ does not have a lattice distribution and
\begin{align}
\frac{H}{h'b_n}\mathrm{Vol}((L_n-v)\cap \{0\}\times [0,b_n))\to 1\n
\end{align}
as $n\to\infty$ uniformly for $v\in \{na'+ih':i\in\bbZ\}$
if $\Vz$ has a lattice distribution with span $h'$ and offset $a'$.
These relations are trivial except for the case that
$V=(\Vz,\Vo)$ has a lattice distribution
and $\Vz$ has a nonlattice distribution,
that is,
$V$ is distributed over
a lattice spanned by $h^{(1)}=(\hzo,\hoo)$ and $h^{(2)}=(\hzt,\hot)$
such that $\hzo/\hzt\notin \bbQ$.
We assume $\hzo>0$ without loss of generality.

In this case
it is necessary to 
evaluate the total number of lattice points in
a rectangle to show \eqref{vol1}.
This number is expressed as
\begin{align}
\lefteqn{
\mathrm{Vol}((L_n-v)\cap [0,\delta)\times [0,b_n))
}\nn
&=
|(L_n-v)\cap [0,\delta)\times [0,b_n)|\nn
&=
\sum_{m=1}^{\lceil \delta/x_1\rceil}
|(L_n-v-((m-1)\delta',0))\cap [0,\delta')\times [0,b_n)|\label{deltap}
\end{align}
for $\delta'=\delta/\lceil \delta/\hzo\rceil<\hzo$.
We can bound each term in \eqref{deltap} by
Lemma \ref{lem_lattice_equi} below, which conclude the proof.
\end{proof}

\begin{lemma}\label{lem_lattice_equi}
Define a rectangle region
$R_n=[0,\delta)\times [0,b_n)$
and
a lattice
$L(a)=\{a+ih^{(1)}+jh^{(2)}:i,j\in\bbZ\}$
for $\hzo>0,\,\hzo/\hzt\notin \bbQ$.
For any $b_n$ such that $\lim_{n\to\infty}b_n=\infty$
and fixed $\delta<\hzo$,
\begin{align}
\frac{|L(a)\cap R_n|}{b_n}
\to \frac{\delta}{H}\n
\end{align}
as $n\to\infty$
uniformly for $a\in\bbR^2$.
\end{lemma}

This lemma intuitively means that the lattice $L(a)$ spanned by
$(h^{(1)},\,h^{(2)})$
contains roughly $\delta b_n/\det (h^{(1)},\,h^{(2)})$
lattice points in a rectangle with size $\delta\times b_n$.
This is intuitively obvious
and the formal proof
of Lemma \ref{lem_blur} is obtained
from the
following proposition.
\begin{proposition}[Equidistribution Theorem]\label{equi}
For any irrational number $\alpha$ and $\delta>0$
it holds that
\begin{align}
\lefteqn{
\limsup_{n\to\infty} \sup_{x\in \bbR}
\frac1n\sum_{i=1}^n
\idx{((x+i\alpha) \mod 1) \in [0,\delta)}
}\nn
&=
\liminf_{n\to\infty} \inf_{x\in \bbR}
\frac1n\sum_{i=1}^n
\idx{((x+i\alpha) \mod 1) \in [0,\delta)}=\delta\per\n
\end{align}
\end{proposition}
This proposition is slightly tighter than
the well-known equidistribution theorem
since the worst-case on $x$ is considered.
We can confirm that the proposition is valid
by following the elementary proof of the equidistribution theorem
in \cite{equidistribution}.

\newcommand{\sx}[1]{x^{(#1)}}
\newcommand{\sy}[1]{y^{(#1)}}
\newcommand{\sz}[1]{z^{(#1)}}
\newcommand{\sS}[1]{S^{(#1)}}
\newcommand{\sL}[1]{L^{(#1)}}
\newcommand{\Jt}[1]{\calJ^{(#1)}}
\newcommand{\uj}[1]{\underline{j}^{(#1)}}
\newcommand{\oj}[1]{\overline{j}^{(#1)}}



\begin{proof}[Proof of Lemma \ref{lem_lattice_equi}]
Define a set of parallel segments $L'(a)=\{L'_j(a): L(a)\cap R_n\neq \emptyset\}$ for 
segment $L_j'(a)=\{a+th^{(1)}+jh^{(2)}:t\in\bbR,j\in\bbZ\} \cap ([0,\delta)\times \bbR)$.
Note that
\begin{align}
\lefteqn{
L_j'(a)
}\nn
&=\brx{a+th^{(1)}+j\pax{\hot-\frac{\hoo \hzt}{\hzo}}\cdot (0,1): t\in\bbR,j\in\bbZ}\nn
&\qquad \cap ([0,\delta)\times \bbR)\per\label{intersect}
\end{align}
Therefore the number of segments $L'_j(a) \in L'(a)$ that intersect with
$[0,\delta)\times \{0\}$ or $[0,\delta)\times \{b_n\}$
is at most
\begin{align}
2\cdot
\left\lceil \delta
\left|
\frac{\hoo}{\hzo}
\Bigg/
\pax{\hot-\frac{\hoo \hzt}{\hzo}}\right|
\right\rceil
=\lo(1)\label{latticepoints}
\end{align}
and we have
\begin{align}
\frac{|L(a)\cap R_n|}{b_n}
&=
\frac{|L(a)\cap L'(a)|+\lo(1)}{b_n}\per\n
\end{align}

Let $\calJ=\{j_0,j_0+1,\cdots,j_1\}$ be the set of indices
of $L_j(a)$ in $L(a)$, that is,
$\calJ$ be such that
$L'(a)=\{L_j'(a):j\in\calJ\}$.
Since
$|L_j'(a) \cap L(a)|\le 1$ holds from $\delta<\hzo$,
we have
\begin{align}
\lefteqn{
|L(a)\cap L'(a)|
}\nn
&=
\sum_{j=j_0}^{j_1}
|L(a)\cap L_j'(a)|\nn
&=
\sum_{j=j_0}^{j_1}
\idx{\exists i\in\bbZ,\,a^{(0)}+i \hzo+j \hzt\in [0,\delta)}
\nn
&=
\sum_{j=j_0}^{j_1}
\idx{((a^{(0)}+j \hzt) \!\!\!\mod{\hzo})\in [0,\delta)}
\nn
&=
(j_1-j_0+1)
\left(
\frac{\delta}{\hzo}+\so(1)\right)\label{apply_equi}
\end{align}
uniformly for $a$ from the equidistribution theorem
in Prop.~\ref{equi}.
From \eqref{intersect}
we have
\begin{align}
(j_1-j_0+1)
&=
|\calJ|=
\frac{b_n+\lo(1)}{\hot-\frac{\hoo \hzt}{\hzo}}.\label{cardi}
\end{align}
Putting \eqref{latticepoints}, \eqref{apply_equi} and \eqref{cardi}
together we have
\begin{align}
\frac{|L(a)-R_n|}{b_n}
&=
\frac{\delta}{\hzo \hot-\hoo \hzt}
+\lo(b_n^{-1})\nn
&=
\frac{\delta}{H}
+\lo(b_n^{-1})\com\n
\end{align}
which concludes the proof.
\end{proof}

\section{Proof of Theorem \ref{thm_equal}}\label{sect_equal}

Since $\rho=1$ and $\eta=1/(1+\rho)=1/2$ for the case of this theorem,
it suffices to show
\begin{align}
\lefteqn{
\E_{\rho}\left[
\e^{-\rho \tZ_1^2/2c_1}
g_{\rho,h}\left(
\frac{
 \e^{\sqrt{n}\tZ_0-n\Delta-\tZ_1^2/2c_1}
}{c_2\sqrt{n}}
\right)
\right]
}\nn
&=
\begin{cases}
\frac{(1+\so(1))h(\e^{h/2}-1)}{4(\e^{h/2}-1)\sqrt{(1+\si_{11}/c_1)}}& \mbox{if $R<\Rcrit$},\\
\frac{(1+\so(1))h(\e^{h/2}-1)}{8(\e^{h/2}-1)\sqrt{(1+\si_{11}/c_1)}}& \mbox{if $R=\Rcrit$}.
\end{cases}\label{suffice_th1}
\end{align}
from discussion around \eqref{tilting}.
Recall that $\Delta>0$ if $R<\Rcrit$ and $\Delta=0$ if $R=\Rcrit$.

Let $W=\tZ_0-(\tZ_1^2/2c_1-c_2\log \sqn)/\sqn$.
Then we have
\begin{align}
\lefteqn{
\E_{\rho}\left[
\e^{-\rho \tZ_1^2/2c_1}
g_{\rho,h}\left(
\frac{
 \e^{\sqrt{n}\tZ_0-n\Delta-\tZ_1^2/2c_1}
}{c_2\sqrt{n}}
\right)
\right]
}\nn
&=
\E_{\rho}\left[
\e^{-\tZ_1^2/2c_1}
g_{1,h}\left(
\e^{\sqn W-n\De}\right)
\right]\!.\n
\end{align}

First we consider the case $R<\Rcrit$.
Since $g_{1,h}(u)\le 1+h/2$
from \eqref{grh}, we have
\begin{align}
&\E_{\rho}\left[
\idx{|\tZ_1|>n^{1/6}}
\e^{-\tZ_1^2/2c_1}
g_{1,h}\left(
\e^{\sqn W-n\Delta}\right)
\right]
=
\e^{-\Omega(n^{1/3})}\per\label{lcrit1}
\end{align}
and
\begin{align}
&
\E_{\rho}\left[
\idx{\tZ_0>n^{1/3}}
\e^{-\tZ_1^2/2c_1}
g_{1,h}\left(
\e^{\sqn W-n\Delta}\right)
\right]
\nn
&
=
\lo(P_{\rho}[\tZ_0>n^{1/3}])\nn
&=\so(1)\per\label{lcrit2}
\end{align}
For the remaining case
we obtain from
$\lim_{u\downarrow 0}g_{1,h}(u)=h(\e^{h/2}+1)/4(\e^{h/2}-1)$ that
\begin{align}
\lefteqn{
\!\!\!\!
\E_{\rho}\left[
\idx{|\tZ_1|\le n^{1/6},\tZ_0\le n^{-1/3}}
\e^{-\tZ_1^2/2c_1}
g_{1,h}\left(
\e^{\sqn W-n\Delta\!}\right)
\right]
}
\nn
&=
\frac{(1+\so(1))h(\e^{h/2}+1)}{4(\e^{h/2}-1)}\nn
&\qquad\qquad\times\E_{\rho}\left[
\idx{|\tZ_1|\le n^{1/6},\tZ_0\le n^{1/3}}
\e^{-\tZ_1^2/2c_1}
\right]\nn
&=
\frac{(1+\so(1))h(\e^{h/2}+1)}{4(\e^{h/2}-1)}
\E_{\rho}\left[
\e^{-\tZ_1^2/2c_1}
\right]\nn
&=
\frac{(1+\so(1))h(\e^{h/2}+1)}{4(\e^{h/2}-1)\sqrt{1+\si_{11}/c_1}}
\com\n
\end{align}
which proves \eqref{suffice_th1}.

Next we consider the case $R=\Rcrit$, where we have $\Delta=0$.
In this case we still have \eqref{lcrit1},
and instead of \eqref{lcrit2} we have
\begin{align}
&\E_{\rho}\left[
\idx{|\tZ_1|\le n^{1/6},|\tZ_0|\le n^{-1/7}}
\e^{-\tZ_1^2/2c_1}
g_{1,h}\left(
\e^{\sqn W}\right)
\right]
\nn
&
=\lo\pax{
P_{\rho}\left[
|\tZ_0|\le n^{-1/7}
\right]
}\nn
&=\so(1)
\n
\end{align}
and
\begin{align}
&\E_{\rho}\left[
\idx{|\tZ_1|\le n^{1/6},\tZ_0> n^{-1/7}}
\e^{-\tZ_1^2/2c_1}
g_{1,h}\left(
\e^{\sqn W}\right)
\right]
\nn
&=
\so(1)
\n
\end{align}
from $\lim_{u\to\infty}g_{1,h}(u)=0$.

For the remaining case
we obtain from
$\lim_{u\downarrow 0}g_{1,h}(u)=h(\e^{h/2}+1)/4(\e^{h/2}-1)$ that
\begin{align}
\lefteqn{
\E_{\rho}\left[
\idx{|\tZ_1|\le n^{1/6},\tZ_0\le -n^{-1/7}}
\e^{-\tZ_1^2/2c_1}
g_{1,h}\left(
\e^{\sqn W}\right)
\right]
}
\nn
&=
\frac{(1+\so(1))h(\e^{h/2}+1)}{4(\e^{h/2}-1)}\nn
&\qquad\qquad\times\E_{\rho}\left[
\idx{|\tZ_1|\le n^{1/6},\tZ_0\le -n^{-1/7}}
\e^{-\tZ_1^2/2c_1}
\right]\nn
&=
\frac{(1+\so(1))h(\e^{h/2}+1)}{4(\e^{h/2}-1)}
\E_{\rho}\left[
\idx{\tZ_0\le -n^{-1/7}}
\e^{-\tZ_1^2/2c_1}
\right]\nn
&=
\frac{(1+\so(1))h(\e^{h/2}+1)}{2(\e^{h/2}-1)}
\E_{\rho}\left[
\idx{\tZ_0\le 0}
\e^{-\tZ_1^2/2c_1}
\right]\per\n
\end{align}
Since region $\{z\in\bbR^2: \idx{z_0\le 0}
\e^{-z_1^2/2c_1}\ge v\}$
is convex for any $v\in\bbR$
we have from multivariate Berry-Esseen bound \cite{berry_multi} that
\begin{align}
\lim_{n\to\infty}
\E_{\rho}\left[
\idx{\tZ_0\le 0}\!
\e^{-\tZ_1^2/2c_1}
\right]
&=
\E_{V\sim \Phi_{\Sigma}}\!\left[
\idx{V_0\le 0}\!
\e^{-V_1^2/2c_1}
\right]\nn
&=
\frac{1}{2\sqrt{1+\si_{11}/c_1}}\per\n
\end{align}
if $|\Sigma|\neq 0$.
It is clear from the one-dimensional Berry-Esseen bound
that the same relation also holds
for the pseudo-symmetric case $|\Sigma|=0$.
\fqed


\section{Proof of Theorem \ref{thm_main}}\label{append_main}
In this appendix we show the main theorem on the exact asymptotics
for the random coding error probability for $R>\Rcrit$.

Define the oscillation of a function $f(z)$ as
\begin{align}
\omega_{f}(S)=
\sup_{z'\in S}f(z')
-
\inf_{z'\in S}f(z')\per\n
\end{align}
Let
\begin{align}
B_n(z)
&=\{z': |z'_0-z_0|\le \de n^{-1/2},\,|z'_1-z_1|\le n^{-1/8}\}\nn
f_n(z)&=
\e^{-\rho z_1^2/2c_1}
g_{\rho,h}\left(
\frac{
 \e^{\sqrt{n}z_0-z_1^2/2c_1}
}{c_2\sqrt{n}}
\right)
\per
\label{def_fn}
\end{align}
Then the oscillation of function $f_n$
is bounded by Lemmas \ref{lem_osci_fn} and \ref{lem_osci} below.
\begin{lemma}\label{lem_osci_fn}
For any $z=(z_0,z_1)$ satisfying
$|z_1|\le M_n=c_1\de n^{1/8}/2$
\begin{align}
\omega_{f_n}(B_n(z))
&\le
4c_h\alpha_{\de}
(\e^{(1-\rho)\sqrt{n}z_0}\land (c_2\sqn)^{\rho}\e^{-\rho\sqn z_0})\com
\n
\end{align}
where $\alpha_{\delta}=\e^{2\delta}-1=\lo(\delta)$.
\end{lemma}

\begin{proof}
For
$z'=(z_0',z_1')\in B_n(z)$ and sufficiently large $n$
we have
\begin{align}
|(z_1')^2-z_1^2|
&\le
|z_1'-z_1|(|z_1'|+|z_1|)\nn
&\le
|z_1'-z_1|(2|z_1|+|z_1-z_1'|)\nn
&\le
n^{-1/4}\left|c_1\delta n^{1/8}+n^{-1/4}\right|\nn
&\le
2c_1\de n^{-1/8}\per\n
\end{align}
For
$w=\sqn z_0-(z_1)^2/2c_1-\log c_2\sqn$
and $w'=\sqn z_0'-(z_1')^2/2c_1-\log c_2\sqn$ we have
\begin{align}
\left|w'-w\right|
&\le
\sqn\left|z_0'-z_0\right|+\frac{|z_1^2-(z'_1)^2|}{2c_1}\nn
&\le
\delta+\delta n^{-1/8}\nn
&\le
2\delta\per\n
\end{align}
Therefore
we obtain for
$\alpha_{\delta}=\e^{2\delta}-1=\lo(\delta)$
and
sufficiently large $n$
that
\begin{align}
\left|\frac{\e^{w'}}{\e^{w}}-1\right|
\le
\alpha_{\delta}\com\qquad
\left|\frac{\e^{\rho\sqrt{n}z_0'}}{\e^{\rho\sqrt{n}z_0}}-1\right|
\le
\alpha_{\delta}\per\n
\end{align}

Now we consider
\begin{align}
f_n(z)=
(c_2\sqn)^{\rho}
\e^{-\rho \sqn z_0}
g_h(\e^w)\per\n
\end{align}
Since $g_h(\cdot)$ satisfies
\begin{align}
|g_h((1+r)u)-g_h(u)|\le c_h|r|(u\land 1)\n
\end{align}
from \cite[Lemma 13]{exact_isit},
it holds for sufficiently large $n$ that
\begin{align}
\frac{f_{n}(z')}{(c_2\sqn)^{\rho}}
&\le
(1+\alpha_{\delta})\e^{-\rho\sqrt{n}z_0}
g_h\left(
(1+\alpha_{\delta})(\e^w\land 1)
\right)\nn
&\le
(1+\alpha_{\de})\e^{-\rho\sqrt{n}z_0}
\left(
g_h\left(
\e^w
\right)
+
c_h\alpha_{\delta}(\e^w\land 1)
\right)\com
\nn
\frac{f_{n}(z')}{(c_2\sqn)^{\rho}}
&\ge
(1-\alpha_{\de})\e^{-\rho\sqrt{n}z_0}
\left(
g_h\left(
\e^w
\right)
-
c_h\alpha_{\delta}(\e^w\land 1)
\right)
.\n
\end{align}
Thus $\omega_{f_n}(B_n(z))$ satisfies
\begin{align}
\lefteqn{
\omega_{f_n}(B_n(z))
}\nn
&\le
2\al_{\de}(c_2\sqn)^{\rho}
\e^{-\rho\sqrt{n}z_0}
\left(
g_h\left(
\e^w
\right)
+
c_h(\e^w \land 1)
\right)\nn
&\le
2\alpha_{\de}(c_2\sqn)^{\rho}
\e^{-\rho\sqrt{n}z_0}
\left(
(c_h\e^w\land 1)
+
c_h(\e^w \land 1)
\right)\nn
&\phantom{wwwwwwwwwwwwwwwwwwww}
\since{by \cite[Lemma 8]{exact_isit}}
\nn
&\le
4c_h\alpha_{\de}(c_2\sqn)^{\rho}
\e^{-\rho\sqrt{n}z_0}
(\e^w \land 1)
\nn
&\le
4c_h\alpha_{\de}
(\e^{(1-\rho)\sqrt{n}z_0}\land (c_2\sqn)^{\rho}\e^{-\rho\sqn z_0})
\com\n
\end{align}
which concludes the proof.
\end{proof}

\begin{lemma}\label{lem_osci}
For $M_n=c_1\delta n^{1/8}/2$,
\begin{align}
\E_{\rho}
\sqx{
\idx{|\tZ_1|\le
M_n
}\omega_{f_n}(B_n(\tZ))
}
&=
\lo(\delta n^{(\rho-1)/2})\per\n
\end{align}
\end{lemma}
\begin{proof}
From Lemma \ref{lem_osci_fn}
we have
\begin{align}
\lefteqn{
\E_{\rho}
\sqx{
\idx{|\tZ_1|\le
M_n
}\omega_{f_n}(B_n(\tZ))
}
}\nn
&\le
4c_h\alpha_{\de}
\Ex{
\idx{\tZ_0< 0}
\e^{(1-\rho)\sqrt{n}Z_0}
}\nn
&\quad+
4c_h\alpha_{\de}
\Ex{\idx{\tZ_0\ge 0}
(c_2\sqn)^{\rho}\e^{-\rho\sqn Z_0}
}\per\label{e_osci_decomp}
\end{align}
Now we show
\begin{align}
\Ex{\idx{\tZ_0\ge 0}
\e^{-\rho\sqn Z_0}
}=\lo(1/\sqn)\per\n
\end{align}
From Berry-Esseen bound
$p_n(z)=\Pr[0\le \tZ_0<z]$ satisfies
\begin{align}
p_n(z)\le \frac{z}{\sqrt{2\pi \sigma_{00}}}+\frac{c}{\sqn}\n
\end{align}
for any $z\ge 0$ and some constant $c>0$.
Thus we have from integration by parts that
\begin{align}
\lefteqn{
\Ex{\idx{\tZ_0\ge 0}
\e^{-\rho\sqn \tZ_0}
}
}\nn
&=
\int_{0}^{\infty} \e^{-\rho \sqn z}\rd p(z)\nn
&=
[\e^{-\rho \sqn z}p(z)]_0^{\infty}
+\rho \sqn
\int_0^{\infty} \e^{-\rho \sqn z}p(z)\rd z\nn
&\le
\rho \sqn
\int_0^{\infty} \e^{-\rho \sqn z}
\pax{\frac{z}{\sqrt{2\pi \si_{00}}}+\frac{c}{\sqn}}\rd z\nn
&=
\pax{\frac{1}{\rho\sqrt{2\pi \si_{00} n}}+\frac{c}{\sqn}}
=\lo(1/\sqn)\per\label{lem4_berry}
\end{align}
The same argument also applies to the first term of \eqref{e_osci_decomp}
and we have
\begin{align}
\E_{\rho}
\sqx{
\idx{|\tZ_1|\le
M_n
}\omega_{f_n}(B_n(\tZ))
}
&=
\lo(\alpha_{\de} n^{(\rho-1)/2})\per\n
\end{align}
We obtain the lemma by recalling that $\alpha_{\de}=\lo(\delta)$.
\end{proof}

Next we show Lemmas
\ref{lem_weight_w} and \ref{lem_continuous} below
on the function $\psi_{\rho,h,h'}$ that was defined by
\begin{align}
\psi_{\rho,h,h'}(x)
&=h'\sum_{i\in\bbZ}g_{\rho,h}(\e^{x+ih'})\nn
&=h'\sum_{i\in\bbZ}\e^{-\rho(x+ih')}g_h(\e^{x+ih'})\per\n
\end{align}

\begin{lemma}\label{lem_weight_w}
For sequences $a_n,b_n>0$ such that $a_n=\so(b_n)$ and $\lim_{n\to\infty}b_n=\infty$,
\begin{align}
\sum_{i=-b_n}^{b_n}
h'g_{\rho,h}(\e^{x+ih'})
-\psi_{\rho,h,h'}(x)
\to 0\n
\end{align}
as $n\to\infty$ uniformly for all $x$ such that $|x|\le a_n$.
\end{lemma}
\begin{proof}
From
\eqref{grh}
we have
\begin{align}
\lefteqn{
\sum_{i=-b_n}^{b_n}
h'g_{\rho,h}(\e^{x+ih'})
-\psi_{\rho,h,h'}(x)
}\nn
&=
\sum_{i=-\infty}^{-b_n}
h'g_{\rho,h}(\e^{x+ih'})
+\sum_{b_n}^{\infty}
h'g_{\rho,h}(\e^{x+ih'})\nn
&\le
c_h h'\left(\sum_{i=-\infty}^{-b_n}
\e^{\rho(x+ih')}
+\sum_{b_n}^{\infty}
\e^{-(1-\rho)(x+ih')}
\right)\nn
&=
c_h h'\left(
\frac{\e^{\rho(x-h'b_n)}}{\e^{-\rho h'}}
+
\frac{\e^{-(1-\rho)(x+h'b_n)}}{\e^{-(1-\rho)h'}}
\right)\nn
&=
c_h h'\left(
\frac{\e^{\rho(a_n-h'b_n)}}{\e^{-\rho h'}}
+
\frac{\e^{-(1-\rho)(-a_n+h'b_n)}}{\e^{-(1-\rho)h'}}
\right)\nn
&=
\so(1)\com\n
\end{align}
where the last equality follows from $a_n=o(b_n)$.
\end{proof}

\begin{lemma}\label{lem_continuous}
$\psi_{\rho,h,h'}(x)$ is Lipschitz continuous in $x\in\bbR$
with a constant independent of $h'$.
\end{lemma}

\begin{proof}
From the periodicity of $\psi_{\rho,h,h'}(x)$ it suffices to consider the case $x \in [0,h')$.
The derivative of each term
of $\psi_{\rho,h,h'}(x)$
is bounded
by
\begin{align}
\lefteqn{
\left|\bibun{\e^{-\rho(x+ih')}g_h(\e^{x+ih'})}{x}\right|
}\nn
&=\left|
-\rho g_{\rho,h}(\e^{x+ih'})
+\e^{-\rho (x+ih')}\e^{x+ih'}\bibun{g_h(u)}{u}\bigg|_{u=\e^{x+ih'}}
\right|\nn
&\le
\rho c_h(\e^{-\rho (x+ih')}\land \e^{(1-\rho)(x+ih')})\nn
&\quad+\e^{(1-\rho)(x+ih')}(\e^{x+ih'}+h\eta)\e^{-\e^{x+ih'}}\nn
&\phantom{wwwwwwwwwwwww}\since{by \cite[Lemma 8]{exact_isit}
and
\eqref{grh}
}\nn
&\le
\rho c_h(\e^{-\rho h' i}\land \e^{(1-\rho)h'(i+1)})\nn
&\quad+\e^{(1-\rho)h'(i+1)}(\e^{h'(i+1)}+h\eta)
(\e^{-(5/2)(x+ih')}\land 1)\com\label{deriv_dec}
\end{align}
where the last inequality follows from
$x\in [0,h)$ and $\e^u \ge (5u/2)\lor 0$ for $u\in \bbR$.

The second term of \eqref{deriv_dec} is bounded by
\begin{align}
\lefteqn{
\e^{(1-\rho)(x+ih')}(\e^{x+ih'}+h\eta)(\e^{-(5/2)(x+ih')}\land 1)
}\nn
&\le
\begin{cases}
(1+h\eta)\e^{(1-\rho)(x+ih')},
&x+ih'<0,\\
(\e^{-(1/2)(x+ih')}+h\eta \e^{-(3/2)(x+ih')}),
&x+ih'\ge 0,\\
\end{cases}\nn
&\le
(1+h\eta)(
\e^{(1-\rho)(i+1)h')}\land
\e^{-ih'/2}
)\per\label{term2}
\end{align}
From \eqref{deriv_dec} and \eqref{term2} we have
\begin{align}
\lefteqn{
\left|\bibun{\e^{-\rho(x+ih')}g_h(\e^{x+ih'})}{x}\right|
}\nn
&\le
\rho c_h(\e^{-\rho ih'}\land c_h \e^{(1-\rho)(i+1)h'})\nn
&\quad+
(1+h\eta)(
\e^{(1-\rho)(i+1)h')}\land
\e^{-ih'/2})\per\label{abs_deriv}
\end{align}
Since the sum of \eqref{abs_deriv} over $i\in \bbZ$ is
convergent,
$\psi_{\rho,h,h'}(x)$ is term-by-term differentiable with
\begin{align}
\left|\bibun{\psi_{\rho,h,h'}(x)}{x}\right|
&\le
\sum_{i\in \bbZ}
h'\rho c_h(\e^{-\rho ih'}\land c_h \e^{(1-\rho)(i+1)h'})\nn
&\quad+
\sum_{i\in\bbZ}
h'(1+h\eta)(
\e^{(1-\rho)(i+1)h')}\land
\e^{-ih'/2})\nn
&=
\lo\left(
\frac{h'}{1-\e^{-h'\min\{\rho,1-\rho,1/2\}}}
\right)\nn
&=\lo(1)\com\n
\end{align}
which implies Lipschitz continuity of $\psi_{\rho,h,h'}$.
\end{proof}

\begin{proof}[Proof Theorem \ref{thm_main}]
First we consider the case that
$(W,R)$ is $(h',a')$-lattice and not pseudo-symmetric.

Let $\calZ_{0,n}=\{(an+ih')/\sqn: i\in\bbZ\}$,
$\calZ_{1,n}=\{i n^{-1/8}: i\in\bbZ,\,|i|\le n^{1/8}M_n\}$,
and
$W=\sqn \tZ_0-\tZ_1^2/2c_1-\log c_2\sqn$.
Then
\begin{align}
\lefteqn{
\E_{\rho}\left[
\e^{-\rho \tZ_1^2/2c_1}
g_{\rho,h}\left(
\frac{
 \e^{\sqrt{n}\tZ_0-\tZ_1^2/2c_1}
}{c_2\sqrt{n}}
\right)
\right]
}\nn
&=
\E_{\rho}\left[
\idx{|\tZ_1|\le M_n}
\e^{-\rho \tZ_1^2/2c_1}
g_{\rho,h}\left(
\frac{
 \e^{\sqrt{n}\tZ_0-\tZ_1^2/2c_1}
}{c_2\sqrt{n}}
\right)
\right]
\nn
&\quad+
\e^{-\Omega(M_n^2)}
\quad\since{by $g_{\rho,h}(\e^w)\le c_h<\infty$ from
\eqref{grh}
}\nn
&=
\sum_{z_0 \in \calZ_{0,n}}
\sum_{z_1 \in \calZ_{1,n}}
\Pr[\tZ_0=z_0,\,\tZ_1- z_1\in [0,n^{-1/8})]\nn
&\phantom{wwwwwwwwww}
\times \e^{-\rho z_1^2/2c_1}
g_{\rho,h}\left(
\frac{
 \e^{\sqrt{n}z_0-z_1^2/2c_1}
}{c_2\sqrt{n}}
\right)\nn
&\quad
+
\lo\pax{
\E_{\rho}\left[
\idx{|\tZ_1|\le M_n}
\omega_{f_n}(B(\tZ))
\right]
}
+
\e^{-\Omega(M_n^2)}
\nn
&=\!
\sum_{z_0 \in \calZ_{0,n}}
\sum_{z_1 \in \calZ_{1,n}}
\!\!
\frac{h'\phi_{\Sigma}(z_0, z_1)}{n^{5/8}}
\e^{-\rho z_1^2/2c_1}
g_{\rho,h}\!\left(
\frac{
 \e^{\sqrt{n}z_0-z_1^2/2c_1}
}{c_2\sqrt{n}}
\right)\nn
&\quad+
\so\!\pax{\frac{1}{n^{5/8}}
\sum_{z_0 \in \calZ_{0,n}}
\sum_{z_1 \in \calZ_{1,n}}
\!\!
\e^{-\rho z_1^2/2c_1}
g_{\rho,h}\!\left(
\frac{
 \e^{\sqrt{n}z_0-z_1^2/2c_1}
}{c_2\sqrt{n}}
\right)
}\nn
&\quad
+
\lo(\de/\sqn)+
\e^{-\Omega(n^{1/4})}\com\label{tobound1}
\end{align}
where the last equality follows from
Lemma \ref{lem_blur} with $b_n:=n^{3/8}$ and Lemma \ref{lem_osci}.
On the second term of \eqref{tobound1}
we can show that
\begin{align}
\sum_{z_0 \in \calZ_{0,n}}
\sum_{z_1 \in \calZ_{1,n}}
\frac{\e^{-\rho z_1^2/2c_1}}{n^{5/8}}
g_{\rho,h}\left(
\frac{
 \e^{\sqrt{n}z_0-z_1^2/2c_1}
}{c_2\sqrt{n}}
\right)=\lo(1/\sqn)\n
\end{align}
in the same way as the evaluation of the first term
of \eqref{tobound1} given below.

We evaluate the first term of \eqref{tobound1}
for $|z_0|>\de_n:=n^{-3/16}$ and $|z_0|\le \de_n$ separately.
For the former case we have
\begin{align}
&\sum_{z_0 \in \calZ_{0,n}: |z_0|>\de_n}
\sum_{z_1\in \calZ_{1,n}}
\!\!\!
\frac{h'\phi_{\Sigma}(z_0, z_1)}{n^{5/8}}\nn
&\phantom{wwwwwwwwwwwwwww}\times
\e^{-\rho z_1^2/2c_1}
g_{\rho,h}\!\left(
\frac{
 \e^{\sqrt{n}z_0-z_1^2/2c_1}
}{c_2\sqrt{n}}
\right)
\nn
&\le
\sum_{z_1\in \calZ_{1,n}}
\sum_{w\in\calW_n(z_1): |w|\ge \sqn \delta_n-M_n^2/2c_1-\log c_2\sqn}
\frac{h'\phi_{\Sigma}(0,0)}{n^{5/8}}\nn
&\phantom{wwwww}
\times
c_h
\e^{-\rho z_1^2/2c_1}
(\e^{-\rho w}\land \e^{(1-\rho)w})
\quad\since{by \eqref{grh}}
\nn
&\le
\sum_{z_1\in \calZ_{1,n}}
\frac{h'\phi_{\Sigma}(0,0)}{n^{5/8}}
\e^{-\rho z_1^2/2c_1}
\cdot \so(1)
\nn
&
=\so(n^{-1/2})\com
\label{delta16}
\end{align}
where
$\calW_n(z_1)=\{na+ih+z_1^2/2c_1-\log c_2\sqn:i\in\bbZ\}$
and \eqref{delta16} follows
from $M_n^2/2c_1+\log c_2\sqn=\lo(n^{1/4})=\so(\sqn \delta_n)$.

On the other hand
for the case $|z_0|\le \delta_n$,
we have
\begin{align}
&\sum_{z_0 \in \calZ_{0,n}: |z_0|\le \delta_n}
\sum_{z_1\in \calZ_{1,n}}
\frac{h'\phi_{\Sigma}(z_0, z_1)}{n^{5/8}}\nn
&\phantom{wwwwwwwwwwwww}\times
\e^{-\rho z_1^2/2c_1}
g_{\rho,h}\!\left(
\frac{
 \e^{\sqrt{n}z_0-z_1^2/2c_1}
}{c_2\sqrt{n}}
\right)
\nn
&=
(1+\so(1))\sum_{z_0 \in \calZ_{0,n}: |z_0|\le \delta_n}
\sum_{z_1\in \calZ_{1,n}}
\frac{h'\phi_{\Sigma}(0, z_1)}{n^{5/8}}
\e^{-\rho z_1^2/2c_1}\nn
&\phantom{wwwwwwwwwwwwwww}
\times g_{\rho,h}\left(
\frac{
 \e^{\sqrt{n}z_0-z_1^2/2c_1}
}{c_2\sqrt{n}}
\right)
\label{z0z1}\\
&=
(1+\so(1))
\sum_{z_1\in \calZ_{1,n}}
\frac{h'\phi_{\Sigma}(0, z_1)}{n^{5/8}}
\e^{-\rho z_1^2/2c_1}\nn
&\qquad
\times
\sum_{w \in \calW_{n}(z_1): |w-z_1^2/2c_1-\log c_2\sqn|/\sqn \le \de_n}
g_{\rho,h}\left(
\e^w
\right)
\nn
&=
(1+\so(1))
\sum_{z_1\in \calZ_{1,n}}
\frac{\psi_{\rho,h,h'}(na-z_1^2/2c_1-\log c_2\sqn)}{2\pi n^{5/8}\sqrt{|\Sigma|}}\nn
&\qquad\times
\e^{-\pax{\frac{\rho}{c_1}+\frac{\si_{00}}{|\Sigma|}}\frac{z_1^2}{2}}\com
\label{use_lem3}
\end{align}
where \eqref{z0z1} and \eqref{use_lem3} follow from $|z_0 z_1|\le \delta_n M_n=\so(1)$
and Lemma \ref{lem_weight_w}, respectively.
Since $\sup_{z_1':|z_1'-z_1|\le n^{-1/8}}((z_1')^2-z_1^2)=\lo(\delta)$ holds
uniformly for $|z_1|\le M_n$, we have from Lemma~\ref{lem_continuous} that
\begin{align}
\lefteqn{
\sum_{z_1\in \calZ_{1,n}}
\frac{\psi_{\rho,h,h'}(na-z_1^2/2c_1-\log c_2\sqn)}{2\pi n^{5/8}\sqrt{|\Sigma|}}
\e^{-\pax{\frac{\rho}{c_1}+\frac{\si_{00}}{|\Sigma|}}\frac{z_1^2}{2}}
}\nn
&=
(1+\lo(\delta))\int_{-M_n}^{M_n}
\frac{\psi_{\rho,h,h'}(na-z_1^2/2c_1-\log c_2\sqn)}{2\pi \sqrt{n|\Sigma|}}\nn
&\phantom{wwwwwwwwwwwwwwwwwwwwww}\times
\e^{-\pax{\frac{\rho}{c_1}+\frac{\si_{00}}{|\Sigma|}}\frac{z_1^2}{2}}
\rd z_1\nn
&=
\int_{-\infty}^{\infty}
\frac{\psi_{\rho,h,h'}(na-z_1^2/2c_1-\log c_2\sqn)}{2\pi \sqrt{n|\Sigma|}}
\e^{-\pax{\frac{\rho}{c_1}+\frac{\si_{00}}{|\Sigma|}}\frac{z_1^2}{2}}
\rd z_1\nn
&\phantom{wwwwwwwwwwwwwwww}+\lo(\delta n^{-1/2})+\e^{-\Omega(M_n^2)}
\nn
&=
\frac{
\E_{V}\sqx{
\psi_{\rho,h,h'}\pax{na'-\frac{|\Sigma|V^2}{2(\sigma_{00}+\rho |\Sigma|/c_1)}-\log c_2\sqn}
}}{\sqrt{2\pi n(\sigma_{00}+\rho|\Sigma|/c_1)}}
\nn&\phantom{wwwwwwwwwwwwwwwwwwwww}
+\lo(\delta n^{-1/2})\label{last_lattice}
\end{align}
and we obtain \eqref{bound_n_l} by letting
$\delta$ sufficiently small.

Next we consider the case that $(W,R)$ is nonlattice.
In this case we replace
\begin{align}
P_{\rho}[\tZ_0=z_0,\,\tZ_1- z_1\!\in\! [0,n^{-1/8})]
\!=\!
\frac{h'\phi_{\Sigma}(z_0,z_1)}{n^{5/8}}+\so(n^{-5/8})\n
\end{align}
with
\begin{align}
&P_{\rho}[\tZ_0-z_0\in[0,\delta n^{-1/2}),\,\tZ_1- z_1\in [0,n^{-1/8})]
\nn
&\phantom{wwwwwwwwwwwwwwwwwww}=
\frac{\delta\phi_{\Sigma}(z_0,z_1)}{n^{5/8}}+\so(n^{-5/8})\label{nonl_thm2p}
\end{align}
by using \eqref{lem_nonlattice} instead of \eqref{lem_lattice}.
To apply \eqref{nonl_thm2p}
we use $\calZ_{0,n}'=\{i\de n^{-1/2}: i\in\bbZ\}$
instead of $\calZ_{0,n}$.
By this change
Eqs.~\eqref{tobound1}--\eqref{last_lattice} are replaced with
\begin{align}
\lefteqn{
\E_{\rho}\left[
\e^{-\rho \tZ_1^2/2c_1}
g_{\rho,h}\left(
\frac{
 \e^{\sqrt{n}\tZ_0-\tZ_1^2/2c_1}
}{c_2\sqrt{n}}
\right)
\right]
}\nn
&=
\frac{
\E_{V}\sqx{
\psi_{\rho,h,\delta}\pax{-\frac{|\Sigma|V^2}{2(\sigma_{00}+\rho |\Sigma|/c_1)}-\log c_2\sqn}
}}{\sqrt{2\pi n(\sigma_{00}+\rho|\Sigma|/c_1)}}
+\lo(\delta n^{-1/2})\n
\end{align}
instead of \eqref{last_lattice}.
Since $\psi_{\rho,h,\delta}(x)$ has period $\delta$,
we obtain from Lemma \ref{lem_continuous} that
\begin{align}
\lefteqn{
\frac{
\E_{V}\sqx{
\psi_{\rho,h,\delta}\pax{-\frac{|\Sigma|V^2}{2(\sigma_{00}+\rho |\Sigma|/c_1)}-\log c_2\sqn}
}}{\sqrt{2\pi n(\sigma_{00}+\rho|\Sigma|/c_1)}}
}\nn
&=
\frac{
\E_{V}\sqx{
\psi_{\rho,h,\delta}\pax{0}+\so(\de)
}}{\sqrt{2\pi n(\sigma_{00}+\rho|\Sigma|/c_1)}}\nn
&=
\frac{
\psi_{\rho,h,\delta}\pax{0}+\so(\de)
}{\sqrt{2\pi n(\sigma_{00}+\rho|\Sigma|/c_1)}}\per\n
\end{align}
We obtain \eqref{bound_n_n} by letting $\delta$ sufficiently small.

Now we consider the case that $(W,R)$ is pseudo-symmetric
and $(h',a')$-lattice.
In this case we have
$Z_1=\sqrt{r}Z_1$
for $r=\sigma_{00}/\sigma_{11}$.
Then, based on the one-dimensional local limit theorem,
\eqref{tobound1} is replaced with
\begin{align}
\lefteqn{
\E_{\rho}\left[
\e^{-\rho \tZ_1^2/2c_1}
g_{\rho,h}\left(
\frac{
 \e^{\sqrt{n}\tZ_0-\tZ_1^2/2c_1}
}{c_2\sqrt{n}}
\right)
\right]
}\nn
&=\!
\sum_{z_0 \in \calZ_{0,n}}
\frac{h'\e^{-z_0^2/\si_{00}}}{\sqrt{2\pi n \si_{00}^2}}
\e^{-\rho rz_0^2/2c_1}
g_{\rho,h}\!\left(
\frac{
 \e^{\sqrt{n}z_0-rz_0^2/2c_1}
}{c_2\sqrt{n}}
\right)\nn
&\quad+
\so\!\pax{\frac{1}{\sqn}
\sum_{z_0 \in \calZ_{0,n}}
\!\!
\e^{-\rho rz_0^2/2c_1}
g_{\rho,h}\!\left(
\frac{
 \e^{\sqrt{n}z_0-rz_0^2/2c_1}
}{c_2\sqrt{n}}
\right)
}\nn
&\quad
+
\lo(\de/\sqn)+
\e^{-\Omega(n^{1/4})}\per\n
\end{align}
By following the argument in \eqref{delta16} and \eqref{use_lem3}
we can ignore $z_0^2$ relative to $\sqn z_0$ and obtain
\begin{align}
\lefteqn{
\E_{\rho}\left[
\e^{-\rho \tZ_1^2/2c_1}
g_{\rho,h}\left(
\frac{
 \e^{\sqrt{n}\tZ_0-\tZ_1^2/2c_1}
}{c_2\sqrt{n}}
\right)
\right]
}\nn
&=\!
(1+\so(1))\sum_{z_0 \in \calZ_{0,n}}
\frac{h'}{\sqrt{2\pi n \si_{00}^2}}
g_{\rho,h}\!\left(
\frac{
 \e^{\sqrt{n}z_0}
}{c_2\sqrt{n}}
\right)
+
\lo(\de/\sqn)\nn
&=
(1+\so(1))
\frac{\psi_{\rho,h,h'}(na'-\log c_2\sqn)}{\sqrt{2\pi n \si_{00}^2}}
+
\lo(\de/\sqn)\per\n
\end{align}
Adaptation of the proof to nonlattice $(W,R)$
is the same as that for the not pseudo-symmetric case.
\end{proof}

\section{Proofs for Singular Channels}\label{sect_proof_sin}
In this section we prove Theorems \ref{thm_equal_sin} and \ref{thm_above_sin}.
The analysis follows the same lines as the analyses for nonsingular
channels in this paper and \cite{exact_isit}
but is much simpler in many places by virtue of the simplicity of the singular channels.

We start with the following lemma
corresponding to Prop.~\ref{expansion_first} and \eqref{tilting} for nonsingular channels.
\begin{lemma}\label{lem_proof_sin}
If channel $W$ is singular then
\begin{align}
\PRC(n)
&=
(1+\so(1))
\e^{-nE(R)}
\E_{\rho}\left[
\gsin_{\rho}\left(
\e^{n(\bZ+R)}
\right)
\right]\per\n
\end{align}
\end{lemma}
\begin{proof}
For the pair of the sent and received sequences $(\bmx,\bmy)$,
the likelihood of the other codeword $\bmX'$
never exceeds that of $\bmx$ and only a tie can occur.
Let $p_0(\bmx,\bmy)$ be the probability that
the likelihood of $\bmX'$ becomes the same as that of $\bmx$
given $(\bmx,\bmy)$, that is,
\begin{align}
p_0(\bmx,\bmy)=\Pr\left[\sum_{i=1}^n \nu(x_i,y_i,X'_i)=0
\right].\n
\end{align}
For probability $p_0=p_0(\bmx,\bmy)>0$ of a tie,
the error probability $q_M=q_M(p_0)$ for $M$ codewords
is expressed as
\begin{align}
q_M(p_0)
&=
1-
\sum_{i=1}^{M-1}\pz^i(1-\pz)^{M-i-1}{{M-1}\choose i}\left(1-\frac{1}{i+1}\right)\nn
&=
1-
\frac{1-(1-p_0)^M}{Mp_0}\label{error_sin_uniform}
\end{align}
by \cite[(23)]{exact_isit}.
An elementary calculation shows
\begin{align}
\limsup_{M\to\infty}\sup_{p_0\in(0,1/2]}
\left|
1-
\frac{q_M(p_0)}{1-\frac{1-\e^{-Mp_0}}{Mp_0}}
\right|
&=
1\com\label{approx_tie_prob}
\end{align}
that is,
$q_M(p_0)$ is uniformly approximated by
$1-(1-\e^{-Mp_0})/(Mp_0)$ with vanishing
relative error for all $p_0 \in (0,1/2]$.
From the definition of nonsingular channels
we have
\begin{align}
p_0(\bmX,\bmY)
&=\prod_{i=1}^n
P_{X'}[\nu(x_i,y_i,X')=0]\nn
&=\prod_{i=1}^n
\e^{Z_i(\eta)}\nn
&=\e^{n\bZ_0}\com\n
\end{align}
where recall that we write $(\bZ_0,\bZ_1)=(\bZ(\eta),\bZ'(\eta))$.
The effect of the case $p_0(\bmX,\bmY)>1/2$ is negligible
by the same argument as the nonsingular channels in \cite[Lemma 5]{exact_isit}.
Thus we obtain from \eqref{approx_tie_prob} that
\begin{align}
\PRC(n)
&=
\E[q_M(\bmX,\bmY)]\nn
&=
(1+\so(1))
\E\left[1-\frac{1-\e^{-\e^{n(\bZ(\eta)+R)}}}{\e^{n(\bZ_0+R)}}\right]\nn
&=
(1+\so(1))
\E\left[
\gsin\left(
\e^{n(\bZ_0+R)}
\right)
\right]\nn
&=
(1+\so(1))
\e^{-nE(R)}
\E_{\rho}\left[
\gsin_{\rho}\left(
\e^{n(\bZ_0+R)}
\right)
\right]\nn
&=
(1+\so(1))
\e^{-nE(R)}
\E_{\rho}\left[
\gsin_{\rho}\left(
\e^{\sqn \tZ_0-n\Delta}
\right)
\right]\com\n
\end{align}
which concludes the proof.
\end{proof}

\begin{proof}[Proof of Theorem \ref{thm_equal_sin}]
First we consider the case
$R<\Rcrit$.
We have $\Delta>0$ and $\rho=1$
in this case
and therefore
\begin{align}
\E_{\rho}\left[
\idx{\tZ_0> n^{1/3}}
\gsin_{1}\left(
\e^{\sqn \tZ_0-n\De}\right)
\right]
&=
\lo(P_{\rho}[\tZ_0> n^{1/3}])\nn
&=\so(1)
\n
\end{align}
since $\gsin_1(u)$ is a bound function.
For the remaining case
we obtain from
$\lim_{u\downarrow 0}\gsin_{1}(u)=1/2$ and the law of large numbers
that
\begin{align}
\lefteqn{
\E_{\rho}\left[
\idx{\tZ_0\le n^{1/3}}
\gsin_{1}\left(
\e^{\sqn \tZ_0-n\Delta}\right)
\right]
}\nn
&=
(1/2+\so(1))
P_{\rho}\left[
\tZ_0\le n^{1/3}
\right]\nn
&=
(1/2+\so(1))\per\label{proof_sin1}
\end{align}

Next we consider the caes $R=\Rcrit$.
In this case we have $\Delta=0$ and $\rho=1$
and therefore
\begin{align}
&\E_{\rho}\left[
\idx{\tZ_0> n^{-1/3}}
\gsin_{1}\left(
\e^{\sqn \tZ_0}\right)
\right]
=
\so(1)
\n
\end{align}
from $\lim_{u\to\infty}\gsin_{1}(u)=0$.
Furthermore,
\begin{align}
&\E_{\rho}\left[
\idx{|\tZ_0|\le n^{-1/3}}
\gsin_{1}\left(
\e^{\sqn \tZ_0}\right)
\right]
\nn
&
=\lo\pax{
P_{\rho}\left[
|\tZ_0|\le n^{-1/3}
\right]
}\nn
&=\so(1)
\n
\end{align}
since $\gsin_{1}(u)$ is a bounded function.
For the remaining case
we obtain from
$\lim_{u\downarrow 0}\gsin_{1}(u)=1/2$ and central limit theorem
that
\begin{align}
\lefteqn{
\E_{\rho}\left[
\idx{\tZ_0< -n^{-1/3}}
\gsin_{1}\left(
\e^{\sqn \tZ_0}\right)
\right]
}
\nn
&=
(1/2+\so(1))
P_{\rho}\left[
\tZ_0< -n^{-1/3}
\right]\nn
&=
(1/2+\so(1))
P_{\rho}\left[
\tZ_0\le 0
\right]\nn
&=
(1/4+\so(1))\com\label{proof_sin2}
\end{align}
We complete the proof by combining
\eqref{proof_sin1} and \eqref{proof_sin2} with
Lemma \ref{lem_proof_sin}.
\end{proof}

Now we move to the proof of Theorem \ref{thm_above_sin}.
We can prove this lemma
by simply replacing the bivariate function $f_n(z)$ given in
\eqref{def_fn} with a univariate function $\gsin_{\rho}(\e^{\sqn z_0})$.
We start with the following bounds on $\gsin_{\rho}(\e^{\sqn z_0})$
to prove counterparts to Lemmas \ref{lem_osci_fn} and \ref{lem_osci}.
\begin{lemma}
\begin{align}
\gsin_{\rho}(\e^{\sqn z_0})&\le
\e^{-\rho\sqn z}\land \e^{(1-\rho)\sqn z},\label{gsin_bound1}\\
\left|\bibun{\gsin_{\rho}(\e^{\sqn z_0})}{z_0}\right|&\le
2\sqn(\e^{-\rho\sqn z}\land \e^{(1-\rho)\sqn z})\per\label{gsin_bound2}
\end{align}
\end{lemma}
\begin{proof}
Eq.~\eqref{gsin_bound1} is straightforward
from
\begin{align}
\gsin_{\rho}(u)
=\frac{1}{u^{\rho}}-\frac{1-\e^{-u}}{u^{1+\rho}}\n
\end{align}
with $\e^{-u}\le 1 \land (1-u+u^2/2)$.

Similarly,
since
\begin{align}
\bibun{\gsin_{\rho}(u)}{u}
&=\rho\frac{1-u-\e^{-u}}{u^{2+\rho}}
+\frac{1-\e^{-u}-u\e^{-u}}{u^{2+\rho}}\com\n
\end{align}
we have
\begin{align}
\left|\bibun{g_{\rho}'(u)}{u}\right|
&\le \frac{-1+u+\e^{-u}}{u^{2+\rho}}
+\frac{1-(1+u)\e^{-u}}{u^{2+\rho}}\n
\end{align}
and we obtain \eqref{gsin_bound2}
by applying $\e^{-u}\le 1 \land (1-u+u^2/2)$
again.
\end{proof}

Now we are ready to prove counterparts to Lemmas \ref{lem_osci_fn} and \ref{lem_osci}.
Let
\begin{align}
B^{(\mathrm{s})}_n(z_0)&=\{z'_0: |z'_0-z_0|\le \de n^{-1/2}\}\com\nn
f^{(\mathrm{s})}_n(z_0)&=\gsin(\e^{-\sqn z_0})\per\n
\end{align}
Then the following lemmas hold.
\begin{lemma}\label{lem_sin_easy}
For $\alpha_{\de}^{(s)}=4\de\e^{\de}=\lo(\de)$,
\begin{align}
\omega_{f_n^{(\mathrm{s})}}(B^{(\mathrm{s})}_n(z_0))\le
\alpha_{\de}^{(s)}(\e^{-\rho\sqn z}\land \e^{(1-\rho)\sqn z})\per\label{osci_sin_eq}
\end{align}
\end{lemma}
\begin{proof}
First we consider the case $z_0\ge 0$.
For $z'_0$ such that $|z'_0-z_0|\le \de n^{-1/2}$,
we have
from \eqref{gsin_bound2} and $\rho\le 1$
that
\begin{align}
f_n^{(\mathrm{s})}(z'_0)
&\le 
f_n^{(\mathrm{s})}(z_0)+
\de n^{-1/2}\cdot 2\sqn \e^{-\rho \sqn (z_0-\de n^{-1/2})}\nn
&\le 
f_n^{(\mathrm{s})}(z_0)+
2\de \e^{\de}\e^{-\rho \sqn z_0}\n
\end{align}
and similarly
\begin{align}
f_n^{(\mathrm{s})}(z'_0)
&\ge 
f_n^{(\mathrm{s})}(z_0)-
2\de \e^{\de}\e^{-\rho \sqn z_0}\com\n
\end{align}
from which \eqref{osci_sin_eq} follows.
The proof for $z_0<0$ is the same as above.
\end{proof}

\begin{lemma}\label{lem_osci_sin}
\begin{align}
\E_{\rho}
\sqx{
\omega_{f_n^{(\mathrm{s})}}(B_n^{(\mathrm{s})}(\tZ_0))
}
&=
\lo(\delta n^{-1/2})\per\n
\end{align}
\end{lemma}

\begin{proof}[Proof of Lemma \ref{lem_osci_sin}]
From Lemma \ref{lem_sin_easy}
we have
\begin{align}
\E_{\rho}
\sqx{
\omega_{f_n}(B_n^{(\mathrm{s})}(\tZ_0))
}
&\le
\alpha_{\de}^{(\mathrm{s})}
\Ex{
\idx{Z_0< 0}
\e^{(1-\rho)\sqrt{n}Z_0}
}\nn
&\quad+
\alpha_{\de}^{(\mathrm{s})}
\Ex{\idx{\tZ_0\ge 0}
\e^{-\rho\sqn \tZ_0}
}.\n
\end{align}
By the same argument as \eqref{lem4_berry}
we have
\begin{align}
&\Ex{\idx{\tZ_0\ge 0}
\e^{-\rho\sqn \tZ_0}
}=\lo(n^{-1/2})\nn
&\Ex{\idx{\tZ_0< 0}
\e^{-(1-\rho)\sqn \tZ_0}
}=\lo(n^{-1/2})\com\n
\end{align}
which conclude the proof.
\end{proof}

\begin{proof}[Proof of Theorem \ref{thm_above_sin}]
Recall that $Z(\eta)$ is not singular
and $\Delta=0$ in this case.

First we consider the case that
$(W,R)$ is $(h',a')$-lattice.
Let $\calZ_{0,n}=\{(an+ih')/\sqn: i\in\bbZ\}$.
Then
\begin{align}
\lefteqn{
\E_{\rho}\left[
\gsin_{\rho}\left(
 \e^{\sqrt{n}\tZ_0}
\right)
\right]
}\nn
&=
\sum_{z_0 \in \calZ_{0,n}}
\Pr[\tZ_0=z_0]
\gsin_{\rho}\left(
 \e^{\sqrt{n}z_0}
\right)\nn
&=\!
\sum_{z_0 \in \calZ_{0,n}}
\!\!
\frac{h'\phi_{\sigma_{00}}(z_0)}{\sqrt{n}}
\gsin_{\rho}\!\left(
 \e^{\sqrt{n}z_0}
\right)
+
\so\!\pax{\!\frac{1}{\sqn}\!
\sum_{z_0 \in \calZ_{0,n}}
\!\!\gsin_{\rho}\!\left(
 \e^{\sqrt{n}z_0}
\right)\!\!
}\!,\label{tobounds1}
\end{align}
where
$\phi_{\si_{00}}$ is the density function of
the normal distribution with zero mean and variance $\si_{00}$
and the last equality follows from
the local limit theorem.

On the second term of \eqref{tobounds1}
we can show that
\begin{align}
\so\!\pax{\frac{1}{\sqn}
\sum_{z_0 \in \calZ_{0,n}}
\gsin_{\rho}\!\left(
 \e^{\sqrt{n}z_0}
\right)
}
=
\so\left(\frac{1}{\sqn}\right)\n
\end{align}
in the same way as the evaluation of the first term
of \eqref{tobounds1} given below.

We evaluate the first term of \eqref{tobounds1}
for $|z_0|>n^{-1/3}$ and $|z_0|\le n^{-1/3}$ separately.
For the former case we have
\begin{align}
\lefteqn{
\sum_{z_0 \in \calZ_{0,n}: |z_0|>n^{-1/3}}
\!\!\!
\frac{h'\phi_{\sigma_{00}}(z_0)}{\sqn}
\gsin_{\rho}\!\left(
 \e^{\sqrt{n}z_0}
\right)
}\nn
&\le
\sum_{z_0 \in \calZ_{0,n}: |z_0|>n^{-1/3}}
\frac{h'\phi_{\sigma_{00}}(0)}{\sqn}\nn
&
\qquad\times(1+h\eta)
(\e^{-\rho \sqn z_0}\land \e^{(1-\rho)\sqn z_1})
\quad\since{by \eqref{gsin_bound1}}
\nn
&=\e^{-\Omega(n^{1/6})}\com
\label{deltas16}
\end{align}

On the other hand
for the case $|z_0|\le n^{1/3}$,
we have
\begin{align}
&\sum_{z_0 \in \calZ_{0,n}: |z_0|\le n^{1/3}}
\frac{h'\phi_{\sigma_{00}}(z_0)}{\sqn}
\gsin_{\rho}\!\left(
 \e^{\sqrt{n}z_0}
\right)
\nn
&=
(1+\so(1))\sum_{z_0 \in \calZ_{0,n}: |z_0|\le n^{1/3}}
\frac{h'\phi_{\sigma_{00}}(0)}{\sqn}
\gsin_{\rho}\left(
 \e^{\sqrt{n}z_0}
\right)\nn
&=
(1+\so(1))
\frac{\psin_{\rho,h'}(na')}{\sqrt{2\pi n}\sigma_{00}}\per
\label{use_lems3}
\end{align}

Next we consider the case that
$(W,R)$ is nonlattice.
In this case
we use $\calZ_{0,n}'=\{i\de/\sqn: i\in\bbZ\}$ instead
of $\calZ_{0,n}$.
By this replacement and local limit theorem for lattice distribution,
we obtain
\begin{align}
\lefteqn{
\E_{\rho}\left[
\gsin_{\rho}\left(
 \e^{\sqrt{n}\tZ_0}
\right)
\right]
}\nn
&=
\sum_{z_0 \in \calZ_{0,n}'}
\Pr[\tZ_0-z_0 \in [0,\de/\sqn)]
\gsin_{\rho}\left(
 \e^{\sqrt{n}z_0}
\right)\nn
&\qquad+
\lo\left(\E_{\rho}[\omega_{f_n^{(\mathrm{s})}}(B_n^{(\mathrm{s})}(\tZ_0))]\right)
\nn
&=\!
\sum_{z_0 \in \calZ_{0,n}'}
\!\!
\frac{h'\phi_{\sigma_{00}}(z_0)}{\sqrt{n}}
\gsin_{\rho}\!\left(
 \e^{\sqrt{n}z_0}
\right)+
\so\!\pax{\!\frac{1}{\sqn}\!
\sum_{z_0 \in \calZ_{0,n}'}
\!\!\gsin_{\rho}\!\left(
 \e^{\sqrt{n}z_0}
\right)\!\!
}
\nn
&\qquad+
\lo\left(\E_{\rho}[\omega_{f_n^{(\mathrm{s})}}(B_n^{(\mathrm{s})}(\tZ_0))]\right)\nn
&=
(1+\so(1))
\frac{\psin_{\rho,\de}(0)}{\sqrt{2\pi n}\sigma_{00}}
+\lo(\de n^{-1/2})\per\n
\end{align}
instead of \eqref{tobounds1}--\eqref{use_lems3}.
We complete the proof by letting $\de$ be sufficiently small.
\end{proof}


\begin{thebibliography}{10}
\providecommand{\url}[1]{#1}
\csname url@samestyle\endcsname
\providecommand{\newblock}{\relax}
\providecommand{\bibinfo}[2]{#2}
\providecommand{\BIBentrySTDinterwordspacing}{\spaceskip=0pt\relax}
\providecommand{\BIBentryALTinterwordstretchfactor}{4}
\providecommand{\BIBentryALTinterwordspacing}{\spaceskip=\fontdimen2\font plus
\BIBentryALTinterwordstretchfactor\fontdimen3\font minus
  \fontdimen4\font\relax}
\providecommand{\BIBforeignlanguage}[2]{{%
\expandafter\ifx\csname l@#1\endcsname\relax
\typeout{** WARNING: IEEEtran.bst: No hyphenation pattern has been}%
\typeout{** loaded for the language `#1'. Using the pattern for}%
\typeout{** the default language instead.}%
\else
\language=\csname l@#1\endcsname
\fi
#2}}
\providecommand{\BIBdecl}{\relax}
\BIBdecl

\bibitem{second_polyanskiy}
Y.~Polyanskiy, H.~Poor, and S.~Verd\'{u}, ``Channel coding rate in the finite
  blocklength regime,'' \emph{IEEE Trans.~Inform.~Theory}, vol.~56, no.~5, pp.
  2307--2359, May 2010.

\bibitem{second_hayasi}
M.~Hayashi, ``Information spectrum approach to second-order coding rate in
  channel coding,'' \emph{IEEE Trans.~Inform.~Theory}, vol.~55, no.~11, pp.
  4947--4966, Nov. 2009.

\bibitem{gallager_map}
R.~G. Gallager, \emph{Information Theory and Reliable Communication}.\hskip 1em
  plus 0.5em minus 0.4em\relax New York: Wiley, 1968.

\bibitem{dyachkov}
A.~D'yachkov, ``Lower bound for ensemble-average error probability for
  a~discrete memoryless channel,'' \emph{Problems of Information Transmission},
  vol.~16, pp. 93--98, 1980.

\bibitem{dobrushin}
R.~L. Dobrushin, ``Asymptotic estimates of the probability of error for
  transmission of messages over a discrete memoryless communication channel
  with a symmetric transition probability matrix,'' \emph{Theory of Probability
  \& Its Applications}, vol.~7, no.~3, pp. 270--300, 1962.

\bibitem{gallager_tight}
R.~G. Gallager, ``The random coding bound is tight for the average code.''
  \emph{IEEE Trans.~Inform.~Theory}, vol.~19, no.~2, pp. 244--246, 1973.

\bibitem{altug_journal}
Y.~Altu\u{g} and A.~Wagner, ``Refinement of the random coding bound,''
  \emph{IEEE Trans.~Inform.~Theory}, vol.~60, no.~10, pp. 6005--6023, Oct 2014.

\bibitem{scarlett}
J.~Scarlett, A.~Martinez, and A.~Guill{\'e}n~i F{\`a}bregas, ``The saddlepoint
  approximation: Unified random coding asymptotics for fixed and varying
  rates,'' in \emph{Proceedings of IEEE International Symposium on Information
  Theory (ISIT14)}, June 2014, pp. 1892--1896.

\bibitem{exact_isit}
\BIBentryALTinterwordspacing
J.~Honda, ``Exact asymptotics for the random coding error probability,'' in
  \emph{Proceedings of IEEE International Symposium on Information Theory
  (ISIT15)}, June 2015, pp. 91--95. [Online]. Available:
  \url{http://arxiv.org/abs/1312.6875}
\BIBentrySTDinterwordspacing

\bibitem{awgn_finite}
T.~Erseghe, ``Coding in the finite-blocklength regime: Bounds based on
  {L}aplace integrals and their asymptotic approximations,'' \emph{IEEE
  Trans.~Inform.~Theory}, vol.~62, no.~12, pp. 6854--6883, 2016.

\bibitem{berry_multi}
R.~N. Bhattacharya, ``{Berry-Esseen} bounds for the multi-dimensional central
  limit theorem,'' \emph{Bulletin of the American Mathematical Society},
  vol.~74, no.~2, pp. 285--287, 1968.

\bibitem{altug}
Y.~Altu\u{g} and A.~Wagner, ``A refinement of the random coding bound,'' in
  \emph{Proceedings of 50th Annual Allerton Conference on Communication,
  Control, and Computing}, Oct 2012, pp. 663--670.

\bibitem{local_bivariate}
R.~A. Doney, ``A bivariate local limit theorem,'' \emph{Journal of Multivariate
  Analysis}, vol.~36, no.~1, pp. 95--102, 1991.

\bibitem{equidistribution}
D.~Speyer, ``Elementary proof of the equidistribution theorem,'' MathOverflow,
  http://mathoverflow.net/q/109158 (version: 2012-10-09).

\end{thebibliography}


\end{document}